\documentclass[11pt]{article}
\usepackage{amsmath,amsfonts,setspace,cancel,url,verbatim,slashed,amsthm,graphicx,color,cite,soul}
\usepackage[a4paper,textwidth=446pt]{geometry}
\newcommand{\be}{\begin{equation}}
\newcommand{\ee}{\end{equation}}

\newtheorem{prop}{Proposition}
\newtheorem{definition}{Definition}

\newcommand{\Ab}{B}
\newcommand{\Ac}{C}
\newcommand{\Ad}{D}
\newcommand{\Ae}{E}
\newcommand{\Af}{F}
\newcommand{\hpartial}{\hat{\partial}}

\usepackage[utf8]{inputenc}

\usepackage[colorlinks=true,citecolor=blue]{hyperref}
\numberwithin{equation}{section}

\author{{\bf Alex S. Arvanitakis} \\ {\tt a.arvanitakis@imperial.ac.uk}\\ \\\small \em The Blackett Laboratory,\\
\small \em Imperial College London, \\
\small \em Prince Consort Road London SW7 2AZ, U.K.}
\title{ \vspace{-3cm}{\hfill \small Imperial-TP-2018-ASA-02}\newline
\bf Brane Wess-Zumino terms from AKSZ and exceptional generalised geometry as an $L_\infty$-algebroid}

\begin{document}

\titlepage

\maketitle

\abstract{We reinterpret the generalised Lie derivative of M-theory $E_6$ generalised geometry as hamiltonian flow on a graded symplectic supermanifold. The hamiltonian acts as the nilpotent derivative of the tensor hierarchy of exceptional field theory. This construction is an M-theory analogue of the Courant algebroid and reveals the $L_\infty$-algebra underlying the tensor hierarchy.

The AKSZ construction identifies that same hamiltonian with the lagrangian of a 7-dimensional generalisation of Chern-Simons theory that reduces to the M5-brane Wess-Zumino term on 5-brane boundaries. The exercise repeats for the type IIB $E_5$ generalised geometry and we discuss the relation to the D3-brane.}

\tableofcontents

\section{Introduction}

In this paper we first clarify certain intricacies in the gauge structure of the ``extended geometries'' relevant for U-duality \cite{Hull:1994ys} in string/M-theory --- meaning both exceptional generalised geometry (EGG) \cite{Hull:2007zu,Pacheco:2008ps,Aldazabal:2013mya} and exceptional field theory (EFT) \cite{
Berman:2010is, Berman:2011jh,Berman:2011cg, Berman:2012vc, Hohm:2013vpa, Hohm:2013uia, Hohm:2014fxa, Hohm:2015xna, Abzalov:2015ega, Musaev:2015ces, Berman:2015rcc} --- and then explain the unexpected connection to brane physics.

It is known that the generalised Lie derivative or Dorfman bracket which contains at the same time both infinitesimal diffeomorphisms and $n$-form gauge transformations of the 10- and 11-dimensional supergravities (theories which are low-energy limits of string/M-theory) fails to satisfy the Jacobi identity, and neither is it antisymmetric; instead, it only satisfies the much weaker identity \eqref{loday}. Algebraically speaking the Dorfman bracket is thus not a Lie algebra bracket. It has been known for a while in the mathematical literature \cite{Baraglia:2011dg,lupercio2014t} that the algebraic structure of the exceptional Dorfman bracket is instead that of an $L_\infty$-algebra \cite{Lada:1992wc}, and while this paper was being written the $L_\infty$-algebra was first related to the ``tensor hierarchy'' of EFT in the physics literature \cite{Cederwall:2018aab}. (For the Courant bracket, relevant for the usual $O(d,d)$ generalised geometry, the $L_\infty$ structure has been known for decades \cite{roytenberg1998courant}, see also \cite{Deser:2016qkw,Hohm:2017pnh}.)

Here, the same $L_\infty$-algebra structure is obtained as the the $L_\infty$-algebra canonically associated \cite{Ritter:2015ffa} to a dg-symplectic manifold $\mathcal M$, (also known as an $L_\infty$-algebroid \cite{Sati:2009ic}, see Definition \ref{dgsymplecticmanifold}) encoding the generalised Lie derivative. What this means is that there is a way to rewrite the exceptional generalised Lie derivatives for both type IIB and M-theory EGGs  as derived Poisson brackets (in the sense of Kosmann-Schwarzbach \cite{kosmann1996poisson} and Voronov \cite{voronov2005higher}) involving a certain odd hamiltonian function $\Theta$ on the symplectic graded manifold $\mathcal M$, such that $(\Theta,(\Theta,\--))=((\Theta,\Theta),\--)/2=0$; in other words the generalised Lie derivative $L_A A'$ is reinterpreted as infinitesimal hamiltonian flow. Explicitly,
\be
L_A A' \sim - \big((\Theta,A),A'\big)\,.
\ee
where sections of the generalised tangent bundle have been identified with functions of some fixed degree on $\mathcal M$.
This is an M-theory/U-duality analogue of the result of Roytenberg \cite{Roytenberg:2002nu,Roytenberg:2006qz} which characterises Courant algebroids \cite{courant1990dirac,liu1995manin,vsevera2015poisson} as dg-symplectic manifolds of degree 2. We find that for the exceptional Dorfman bracket one instead has dg-symplectic manifolds of degrees 6 or 4, depending on whether we consider an M-theory or type IIB kind of construction respectively.

The reason this is interesting is that \emph{the only input that goes into these constructions is the content of the exceptional tangent bundle of EGG}, i.e. just the content of $R_1$ from Table \ref{udualitytable}; once that is known, the candidate expression for the Dorfman bracket is fixed (in the absence of twists). After checking the correctness of the Dorfman bracket thereby defined (in Propositions \ref{genLieMtheory} and \ref{genLieIIB}), we obtain from known mathematical results \cite{Ritter:2015ffa,getzler2010higher,fiorenza2007structures} the $L_\infty$-algebra structure of the Dorfman bracket, for which the ``lowest'' component consists of generalised vectors, and whose higher components are explicitly shown to correspond to the higher modules $R_2,R_3\dots$ forming the ``tensor hierarchy'' \cite{Aldazabal:2013via,Berman:2012vc,Hohm:2015xna,Wang:2015hca} of EFT. The hamiltonian vector field $(\Theta,\--)$ is simply the well-known EFT nilpotent derivative operator $\hat \partial$ \cite{Cederwall:2013naa}, squaring to zero because $(\Theta,\Theta)=0$ (not an identity, because $\Theta$ is odd). As a further application, we classified all possible twists for the Dorfman bracket simply by writing down the most general expression for $\Theta$ consistent with $(\Theta,\Theta)=0$. Given the relation between $\Theta $ and the EFT derivative $\hat \partial$, this illustrates the hitherto-unexplored possibility of twisting the EFT tensor hierarchy chain complex.

 Furthermore we find that the same odd hamiltonians can be connected through the AKSZ construction \cite{alexandrov1997geometry} to the physics of certain branes: they turn out to define a ``topological'' subsector of the M5-brane \cite{Townsend:1995af} and D3-brane \cite{Tseytlin:1996it} lagrangians, specifically the Wess-Zumino (WZ) term that describes the coupling to a supergravity gauge field background. 
  The relation of those AKSZ topological field theories to brane physics is a direct generalisation of the known relation of the ``Courant'' sigma model (named for the relation to the Courant bracket, but actually written down by Ikeda \cite{Ikeda:2002wh}) to the string WZ term.

 That physical interpretation involves putting the respective AKSZ sigma models on 7- and 5-dimensional manifolds with boundary, which is identified with the brane worldvolume. Careful examination of consistent boundary conditions for the variational problem leads near-inevitably to the introduction of the correct worldvolume gauge fields, including the chiral 2-form on the M5, although it must be said that the M5-brane case is rather more clear-cut than the D3-brane one. Remarkably, one form of the M5-brane WZ term we derive has been found before by Kalkkinnen and Stelle \cite{Kalkkinen:2002tk} from a careful analysis of large gauge transformations in M-theory. It is in turn also related to the ``Hopf-WZ'' term of Intriligator \cite{Intriligator:2000eq}, introduced on the worldvolume of a probe M5 in an M5-brane background from anomaly-matching considerations.

 Before we begin reviewing material we mention some work which is perhaps not a direct ancestor of this paper but is nevertheless relevant. Tensor hierarchies (different from the EFT one considered here) have been connected to $L_\infty$-algebras before, by at least two different groups \cite{Palmer:2013pka,Lavau:2014iva}. There have been a number of papers proposing derived bracket structures for DFT \cite{Deser:2014mxa,Deser:2016qkw,Heller:2016abk}, two for EGG \cite{lupercio2014t} \cite{Baraglia:2011dg} and  one for EFT \cite{Cederwall:2018aab}; in the last two one also finds the EGG/EFT $L_\infty$-algebra structure respectively in the M-theory case (although none exhibit the $L_\infty$-\emph{algebroid} or dg-symplectic structure which is the main point of this paper). On topological sigma models in the context of extended geometries we mention the works \cite{Heller:2016abk,Chatzistavrakidis:2018ztm,Chatzistavrakidis:2015vka} in which generalisations of the Courant sigma model are used to study non-geometry in string theory (which strongly suggest appropriate generalisations of the theories developed in this paper would be relevant for non-geometry in M-theory), and also \cite{Fiorenza:2013nha} wherein WZ terms for various branes (including the M5) are found from $L_\infty$-algebraic generalisations of the super-Poincar\'e algebra (in a manner reminiscent of \cite{Townsend:1997wg} and references therein). 

Finally, after this paper first appeared online we were informed that the $E_6$ construction had been written down before, in unpublished work \cite{DeserSaemannUnpublushed}.

\subsection{Exceptional generalised geometry and exceptional field theory}
Let us recall the setting of extended geometries in the above sense. Ten-dimensional type IIB and the eleven-dimensional supergravity describe the physics of IIB string theory and M-theory respectively in a low-energy limit. In both supergravity theories there are of course metric degrees of freedom but also various kinds of matter degrees of freedom; besides the fermion sector (which we will not deal with in any way here) there is a menagerie of other fields: the NSNS and RR gauge fields in IIB and the 3-form in 11D. In lower-dimensional supergravities one gets even more fields. For those lower-dimensional supergravities arising from compactification of 11D supergravity on a $d$-torus, these fields organise themselves into representations of the U-duality groups \cite{Cremmer:1979up}. Extended geometry can be somewhat loosely defined as the programme where all of these matter degrees of freedom as well as the metric are collectively described in terms of a bigger ``generalised metric'' structure, and the combined infinitesimal diffeomorphisms and $q$-form gauge transformations are compactly encoded inside a ``generalised Lie derivative'' (sometimes also known as a ``Dorfman bracket''). Merits of the programme  include
\begin{enumerate}
	\item that it provides more powerful formulations of supergravity backgrounds: supersymmetry conditions are most naturally expressed in this context, and have suggested e.g. vast generalisations of Calabi-Yau manifolds \cite{Hitchin:2004ut,Ashmore:2015joa};
	\item and that it generally leads to (both T- and U-) duality-invariant descriptions for physics, with obvious applications; for instance, to ``T-folds'' \cite{Hull:2004in,Hull:2006va} i.e. locally or globally nongeometric spaces with T-duality transition functions, and to the ``exotic branes'' that often give rise to them (see e.g. \cite{deBoer:2010ud}).
\end{enumerate}

Actually, each of the two applications happens to be associated with the each of the two approaches to extended geometry for U-duality we mentioned previously: EGG and EFT respectively. They each in turn generalise the corresponding notions for T-duality: generalised geometry \cite{Hitchin:2004ut,Gualtieri:2003dx} and double field theory (DFT) \cite{Hull:2009mi}). We will have to consider all of those but will generally focus on the exceptional constructions.

Take for definiteness an 11D supergravity background of the form
\be
P\times M\,.
\ee We will mostly concern ourselves with the ``internal'' manifold $M$ with $\dim M=d$. This will be the (compact) space along which T- and U-dualities act. $P$ is the ``external'' space of dimension $n$ ($n+d=11$ for M-theory), spectating the dualities. In the EGG construction for M-theory (i.e. for 11D supergravity) one extends the tangent bundle $TM$ to the ``generalised tangent bundle''
\be
E:=TM \oplus\Lambda^2 T^\star M \oplus \Lambda^5 T^\star M\,,
\ee whose ``generalised vector'' sections include a diffeomorphism generator $v \in \Gamma(TM)$ as well as gauge parameters $\omega \in \Gamma(\Lambda^2 T^\star M),\sigma \in \Gamma(\Lambda^5 T^\star M)$ for the 3-form potential and its magnetic dual. For $2\leq d \leq 6$ the fibres of $E$ form a certain representation $R_1$ of $E_d\times \mathbb{R}^+$, where $E_{d\geq 6}$ is the split real form of the exceptional Lie group $E_d$ and is otherwise given by the table; ``generalised tensors'' are  sections of vector bundles carrying other representations of $E_d\times \mathbb{R}^+$, and generalised vectors act on them by generalised Lie derivatives. In EGG one discards $P$ and considers $M$ alone, which corresponds to a subsector of the fields in 11D supergravity.

\begin{table}[h]\centering
\begin{tabular}{|c|c|c|c|c|c|c|c|c|}
\hline
$n = 11-d$ & $d$ &  $ E_d$ & $R_1$ & $R_2$ &$R_3$ &$R_4$ &$R_5$ &$R_6$ \\ \hline
9 & 2 & $SL(2)\times \mathbb{R}^+$  & $\mathbf{2_1} \oplus \mathbf{1_{-1}}$  & $ \mathbf{2_0} $  & $\mathbf{1_1}$ & $\mathbf{1_0}$ & $\mathbf{2_1}$ & $\mathbf{2_0}\oplus \mathbf{1_2}$   \\
8 & 3 & $SL(3)\times SL(2)$   & $ (\mathbf{3,2})$ & $(\mathbf{\bar 3,1})$  &$\mathbf{(1,2)}$ &$\mathbf{(3,1)}$ &$\mathbf{(\bar 3,2)}$ &$\ast$  \\
7 & 4 &  $SL(5)$  & $ \mathbf{10}$& $\mathbf{\bar{5}}$ &$\mathbf{5}$& $\mathbf{\overline{10}}$  &$\ast$ \\
6 & 5 & ${\rm Spin}(5,5)$   & $\mathbf{16}$ & $\mathbf{10}$ &$\mathbf{\overline{16}}$ &$\ast$  \\
5 & 6 & $E_6$   & $\mathbf{27}$& $\mathbf{\overline{27}}$ &$\ast$ \\
4 & 7 &  $E_7$   & $\mathbf{56}$ & $\ast$  \\
3 & 8 & $E_8$  & $\mathbf{248}$   \\
\hline
\end{tabular}
\caption{U-duality groups and corresponding EFT tensor hierarchy representations. $n$ is the external dimension while $d$ is the internal dimension of M-theory sections (so $d-1$ is the internal dimension for IIB sections, and $\dim R_1$ is the dimension of the extended internal space $\tilde M$ of EFT). $\ast$ marks the end of the tensor hierarchy.}
\label{udualitytable}
\end{table}
EFT on the other hand goes further in two ways: first by extending the base, replacing $M$ with a bigger manifold $\tilde M$ (such that $T\tilde M = E$, loosely speaking, implying $\dim \tilde M=\dim R_1$), and then by considering both $\tilde M$ and $P$ at the same time, providing a reformulation of the full 11D supergravity. Consistency of the theory however implies that one must constrain all fields to depend on at most $d$ coordinates (``(strong) section condition''). Once a section is chosen and identified with $M$ one makes contact with the setting of EGG as described above. Sections come in two inequivalent types: if we fix the U-duality group to $E_d$, M-theory sections are $d$-dimensional and lead to the M-theory EGG construction described in the previous paragraph, while type IIB sections are $(d-1)$-dimensional and lead to the IIB EGG construction. A remaining difference between EFT and EGG however is that because EFT takes the external space $P$ into account, there appear a series of fields --- $q$-forms on $P$ taking values in ``generalised tensors'' on $M$ transforming in the $R_q$ representation of $E_d$\footnote{\emph{Except} for $q=1$: what is usually called the $R_1$ field is not a generalised tensor; rather, it transforms as a Yang-Mills gauge potential on external space with ``gauge group'' the generalised diffeomorphisms of internal space. We will not continue to point out this subtlety later on.} --- which have hitherto not played a role in EGG. Together, these form the ``tensor hierarchy'' \cite{Berman:2012vc,Hohm:2015xna}. The relevant representations $R_q$ are also listed in the table.

The EFT tensor hierarchy considered here is specifically the one which forms a chain complex under the connection-free EFT nilpotent derivative operator $\hat \partial$ \cite{Cederwall:2013naa} ($\hat \partial^2=0$):
\be
R_1 \overset{\hat \partial}{\leftarrow} R_2\overset{\hat \partial}{\leftarrow} \dots \overset{\hat \partial}{\leftarrow} R_{8-d}\qquad (2\leq d \leq 6)\,.
\ee
$\hat \partial$ between these modules is covariant under the generalised Lie derivative and is defined without the introduction of a connection, which is analogous to how the exterior derivative $d$ is covariant under the usual Lie derivative. More precisely this is true if the $\mathbb{R}^+$ weight (in $E_d\times \mathbb{R}^+$) of the field in $R_q$ is \cite{Wang:2015hca}
\be
\label{godweight}
\frac{q}{d-2}=:-q\omega\,.
\ee

It is regrettable we had to introduce so much jargon in so little a section, however one of the aims of this paper is to clarify how a number of these notions fit together and thus it is unavoidable to refer to these ideas. We refer to the original EFT and EGG literature for less terse explanations. Notation and terminology is generally uniform within each. We will use the conventions of \cite{Ashmore:2015joa} for EGG (especially appendix E) and \cite{Arvanitakis:2018hfn} and \cite{Arvanitakis:2016zes} for EFT and DFT respectively.

\subsection{An $O(d,d)$ story}
It is instructive to summarise the story we will describe in the rest of the paper in the technically simpler context of T-duality, ordinary (as opposed to exceptional) generalised geometry, double field theory, and Courant algebroids. The relation to graded symplectic supermanifolds and $L_\infty$-algebras is not original (see e.g. the review \cite{Roytenberg:2006qz}) but the link to the tensor hierarchy in double field theory \cite{Hohm:2013nja} is.

Generalised geometry in the sense of \cite{Hitchin:2004ut,Gualtieri:2003dx} involves extending the tangent bundle $TM$ of a $d$-dimensional manifold $M$ to the generalised tangent bundle $TM\oplus T^\star M$. Generalised vectors are sections $A$ of $TM\oplus T^\star M$. They encode infinitesimal diffeomorphisms on $M$ as well as the 1-form gauge parameters for the (locally-defined) string theory NSNS $B$-field with closed 3-form field strength $H\in \Lambda^3 T^\star M$. The generalised Lie derivative or Dorfman bracket of two sections $A,A'$ of $TM\oplus T^\star M$ is another section
\be
\label{genlieodd}
L_A A':= (\mathcal L_v v', \mathcal L_v \lambda'- \iota_{v'} d\lambda)\qquad (A=(v,\lambda)\,, v\in \Gamma(TM)\,, \lambda \in \Gamma(T^\star M))
\ee
where $\mathcal L_v$ is the ordinary Lie derivative or Lie bracket of vector fields ($\mathcal L_v v':=[v,v']$) and $\iota_v$ is the contraction or interior derivative with $v$. It satisfies a not-quite-Jacobi identity
\be
\label{loday}
L_{A_1} L_{A_2} A_3=L_{L_{A_1} A_2} A_3+ L_{A_2} L_{A_1} A_3
\ee
and thus the space of generalised vectors forms a ``Leibniz'' or ``Loday algebroid'' (see e.g. \cite{Baraglia:2011dg}). This identity is \emph{not} equivalent to the Jacobi identity because $L_{A_1}A_2$ is not antisymmetric in $1\leftrightarrow 2$; rather, we have
\be
\label{symmgenlieodd}
L_{A_1} A_2 + L_{A_2} A_1=d(\iota_{v_1} \lambda_2 + \iota_{v_2}\lambda_1)\,.
\ee
In the right-hand side we encounter the natural $O(d,d)$-invariant metric $\langle \--,\--\rangle:T_xM\oplus T_x^\star M\to \mathbb{R}$ on the fibres of $TM\oplus T^\star M$ at each point $x \in M$. Along with the ``Courant bracket'' $[A,A']_C:=(L_A A'-L_{A'} A)/2$ and the ``anchor'' $a:TM\oplus T^\star M \to TM$ it makes $(M,TM\oplus T^\star M,[\--,\--]_C,\langle \--,\--\rangle,a)$ into a ``Courant algebroid'' as introduced in \cite{liu1995manin}. It is easy to show using the axioms defining Courant algebroids that one can always obtain a Dorfman bracket satisfting \eqref{loday} from the Courant bracket.

The above example is in fact an exact Courant algebroid in the sense that the sequence
\be
0 \to T^\star M\overset{a^T}{\to} TM\oplus T^\star M\overset{a}{\to} TM \to 0
\ee
is an exact sequence of vector bundles. Exact Courant algebroids are classified by the third de Rham cohomology class of $M$ \cite{vsevera2015poisson}. They are all of the above form up to isomorphism except for the Courant and Dorfman brackets: an isomorphism singles out a closed 3-form $H$ representing the class in $H^3(M,\mathbb{R})$ and the generalised Lie derivative turns into the ``twisted'' generalised Lie derivative
\be
L_A A':= (\mathcal L_v v', \mathcal L_v \lambda'- \iota_{v'} d\lambda+\iota_{v'}\iota_{v} H)\,.
\ee
One also says the generalised Lie derivative is twisted by $H$. The Jacobi-esque identity \eqref{loday} is satisfied by virtue of $dH=0$, as are the Courant algebroid axioms.

In the double field theory picture \cite{Hull:2009mi} one has the same story as far as the purely internal sector is concerned except the role of the T-duality group $O(d,d)$ is front-and-centre. One can rewrite the untwisted generalised Lie derivative in terms of $O(d,d)$ tensors. First let $M=1,2,\dots 2d$ denote a fundamental $O(d,d)$ index. Generalised vectors are thus written $A^M$. The $O(d,d)$-invariant metric on $TM\oplus T^\star M$ is denoted $\eta_{MN}$ and 
\be
A_1^M A_2^N\eta_{MN}=\langle A_1,A_2\rangle\,, \qquad \eta_{MN}=\begin{pmatrix}0_d& 1_d\\ 1_d & 0_d \end{pmatrix}\,,\quad A^M=(v^\mu,\lambda_\nu)\,.
\ee
The anchor map $a$ in this notation is simply $a^\mu_M A^M=v^\mu$. One can write $\partial_M:= a_M^\mu\partial_\mu$ in terms of which the untwisted\footnote{In the presence of a nontrivial twist $H$, one can still use the generalised Lie derivative in this untwisted form by instead twisting the generalised vectors by a gerbe \cite{Hull:2014mxa}. This is convenient in EFT because the twists haven't yet been written in an $E_d$ covariant form. Untwisted vectors have twisted Lie derivatives and vice versa. For DFT on the other hand a proposal for the twisted D-bracket is in section 6.4 of Deser and S\"amann \cite{Deser:2016qkw}.} generalised Lie derivative is
\be
(L_A A')^M=  A^N\partial_N A'\,^M-A'\,^N \partial_N A^M + \eta^{MN}\eta_{PQ}\partial_N A^P A'\,^Q \,.
\ee
The identity \eqref{symmgenlieodd} is now
\be
\label{symmgenlieodddft}
(L_A A')^M +(L_{A'} A)^M=\eta^{MN}\partial_N(\eta_{PQ} A^P A'\,^Q)
\ee
which will be important in the following.

In DFT one now proceeds by ``forgetting'' which manifold $M$ the vector bundle $TM\oplus T^\star M$ arose from, so $\partial_M$ is formally a partial derivative on a doubled space $\tilde M$ of dimension $2d$. Consistency, including the identity \eqref{loday}, requires the ``section condition'' or ``strong constraint''
\be
\eta^{MN}\partial_M\otimes \partial_N=0
\ee
so $\partial_M$ lies on a null subspace for the $O(d,d)$ metric $\eta_{MN}$ at each point. As these are always $d$-dimensional, at most $d$ components of $\partial_M$ are nonzero and generalised tensors are interpreted as (local) sections of ordinary vector bundles on that $d$-dimensional manifold. One thus views the generalised geometry construction on $M$ described above as this DFT construction on $\tilde M$ with a preferred $d$-dimensional section $M$. This only makes sense locally --- $\tilde M$ currently lacks a sensible description outside of local coordinate patches --- but for a lot of physical applications that is good enough.

While one can reformulate the common or purely NS sector of both 10-dimensional type II supergravities as the $O(10,10)$ DFT, to make touch with the exceptional field theory construction we will only double $d<10$ coordinates as in \cite{Hohm:2013nja}. The 10-dimensional spacetime is thus written as $P\times M$ with external space $P$ and internal space $M$ (of dimension $d$), or $P\times \tilde M$ in the DFT picture. In this case it is easy to handwave where the tensor hierarchy comes from. Writing the 10-dimensional dynamics in a ``Kaluza-Klein-esque'' split $P\times M$ (but without reducing on $M$) produces $d$ 1-form gauge fields $A_i{}^\mu$ on $P$ from the Kaluza-Klein ansatz on the 10-dimensional metric. These are completed with modes arising from the $B$-field into the $2d$ 1-forms $A_\mu{}^M$ on $P$. As we have not actually reduced on $M$, $\partial_M \neq 0$ and the failure of the Jacobi identity for the Courant bracket implies that the obvious field strength
\be
2 \partial_{[i} A_{j]}^M - [A_i,A_j]^M_C
\ee
will fail to be a generalised tensor. As identity \eqref{symmgenlieodddft} suggests, this failure is by terms of the form $\eta^{MN}\partial_N(\cdots)$ and can be cured by introducing a term $\eta^{MN}\partial_N B_{ij}$ to the field strength. As it happens, the 2-form $B_{ij}$ (which is indeed related to the 10-dimensional string theory $B$-field) has a perfectly sensible field strength without the introduction of additional fields. The $O(d,d)$ representations in the tensor hierarchy are thus $R_1=\mathbf{2d}, R_2=\mathbf{1}$ and the hierarchy reads
\be
\mathbf{2d} \overset{\hat \partial}{\leftarrow} \mathbf{1}\,, \qquad (\hat \partial B)^M:=\eta^{MN}\partial_N B
\ee
where we have read off the form of the ``nilpotent'' derivative $\hat \partial$, whose EFT counterpart is genuinely nilpotent (we dropped the external space $P$ indices $ij$).

To complete the analogy with the EFT tensor hierarchy, introduce the product $ A_1\bullet A_2= A_1^M A_2^N\eta_{MN}$ mapping $R_1\otimes R_1$ into $R_2$. The right-hand side of the symmetric part of the generalised Lie derivative \eqref{symmgenlieodddft} is then $\hat\partial(A\bullet A')$ and the  generalised Lie derivative $L_A B$ is simply $L_A B=A \bullet\hat \partial B$. If we also define $A\bullet B=B\bullet A=0$ we can write this in a form reminiscent of Cartan's magic formula:
\be
\label{oddmagicformula}
L_A B=A \bullet\hat \partial B + \hat\partial(A\bullet B)\,.
\ee

It is striking here that while the derivation of the tensor hierarchy in this case was motivated by considerations involving the external space $P$, the tensor hierarchy itself seems to be independent of the details of $P$. One could have arrived at the same result by considering how the generalised Lie derivative should act on the $O(d,d)$ invariant $A_1^M A_2^N\eta_{MN}$ and writing that in terms of a magic formula. Indeed the EFT tensor hierarchy can and has been derived this way (as in e.g. \cite{Hohm:2015xna}). Given that, it is no surprise we will find a way to derive the same using only the generalised Lie derivative, but what we will do is fully systematic and yields more information --- the $L_\infty$-algebra structure, and the canonically associated topological field theory, to name just two things.

\subsubsection{The dg-symplectic geometry picture, $L_\infty$, AKSZ and the string WZ coupling}

One can derive all of the above in a near-mechanical manner by first reformulating the generalised Lie derivative as a derived bracket on a certain graded symplectic supermanifold. This construction is identical to the one of \cite{Roytenberg:2006qz}.  Let $x^\mu$ be a local coordinate on $M$, then take the graded supermanifold $\mathcal M$ with local homogeneous coordinates:
\begin{align}
\label{oddgradedmanifolddegrees_odd}
\begin{matrix}
&(x^\mu\,, &\psi^\mu\,, &\chi_\mu\,, &p_\mu)\,. \\
\deg&0 &1 &1 &2
\end{matrix}
\end{align}
$\deg$ denotes the degree of a local coordinate, and the degree assignments are compatible with supermanifold parity in the sense $ab=(-1)^{(\deg a)(\deg b)} ba$. (Appendix \ref{appendix:gradedsupermanifolds} has details on definitions and conventions for graded symplectic supermanifolds.)

The space of functions $C^\infty(\mathcal M)$ is a direct sum of functions with definite degrees, some of which can be identified with vector bundles on $M$.  At degree zero we simply have $C^\infty(M)$, while the general function $A$ of degree 1 is
\be
A=v^\mu(x) \chi_\mu + \lambda_\mu(x) \psi^\mu
\ee
and is thus identified with a section of $TM\oplus T^\star M$, i.e. a generalised vector. If we then assign the Poisson brackets
\be
(x^\mu,p_\nu)=-(p_\nu,x^\mu)=\delta_\mu^\nu\,,\qquad (\psi^\mu,\chi_\nu)=(\chi_\nu,\psi^\mu)=\delta_\mu^\nu
\ee
a short calculation using the graded Leibniz rule for the Poisson bracket implies that the generalised Lie derivative $L_A A'$ \eqref{genlieodd} can be obtained from the following derived bracket expression with hamiltonian $\Theta:=p_\mu \psi^\mu$ of degree 3 (we use the same symbols for generalised vectors and the corresponding functions on $C^\infty (\mathcal M)$):
\be
-((\Theta,A),A')=\left(v^\mu\partial_\mu v'\,^\nu-v'\,^\mu\partial_\mu v^\nu\right)\chi_\nu + \left(v^\mu\partial_\mu\lambda'_\nu+\partial_\nu v^\mu\lambda'_\mu -2v'\,^\mu \partial_{[\mu}\lambda_{\nu]} \right)\psi^\nu\,.
\ee
In terms of the hamiltonian vector field $X_f=(f,\--)$ for any function $f$ in $C^\infty(\mathcal M)$ \eqref{poissondef} we express $L_A A'$ compactly as $ -X_{(\Theta,A)}\cdot A'$.

The identities we previously displayed can now be obtained using the graded Jacobi identity \eqref{poissonjacobi} of the Poisson bracket as well as its commutativity properties. For instance the symmetrised derivative identity \eqref{symmgenlieodd} is equivalent to
\be
((\Theta,A),A')+((\Theta,A'),A)=(\Theta,(A,A'))\,.
\ee
Less trivially and more importantly, the not-quite-Jacobi identity \eqref{loday} defining the Leibniz algebroid structure of the generalised Lie derivative of generalised vectors follows from \eqref{pbhomomorphism} and
\begin{align}
\label{lodayproof1}
((\Theta,A_1),(\Theta,A_2))&=&\Big(\big((\Theta,A_1),\Theta\big),A_2\Big) +&\Big(\Theta,\big((\Theta,A_1),A_2\big)\Big)\\
\label{lodayproof2}
&=&\Big(\big((\Theta,A_1),\Theta\big),A_2\Big) +&\Big(\Theta,\big(\Theta,(A_1,A_2)\big)\Big) -\Big(\Theta,\big((\Theta,A_2),A_1\big)\Big)\,.
\end{align}
The last term acts on $A_3$ as $-L_{L_{A_2}A_1} A_3=+L_{L_{A_1}A_2} A_3$ while the first two terms will vanish if $X_\Theta^2=0$. Given that $\deg \Theta=3=2+1$ this in turn follows from \eqref{pbhomomorphism}:
\be
X_\Theta^2=\frac{1}{2} X_{(\Theta,\Theta)}\,.
\ee
Since $\Theta=p_\mu \psi^\mu$ we trivially have $(\Theta,\Theta)=0$.

If $B$ is a function of degree 0 we also recover the anchor map as $-X_{(\Theta,A)}\cdot B$. This can be rewritten in a form analogous to the magic formula \eqref{oddmagicformula}:
\be
-((\Theta,A),B)=-((\Theta,B),A)
\ee
and we observe the direct correspondence between the bullet product $\bullet$ and the Poisson bracket and also between $\hat \partial$ and $X_\Theta$ which will persist in the exceptional case.

\bigskip

We have therefore encoded the Dorfman bracket and anchor in terms of a graded symplectic supermanifold $\mathcal M$ with a compatible ``homological'' vector field $X_\Theta$ (i.e. $\deg X_\Theta=1$ ands $X_\Theta^2=0$), or a ``${\rm dg}$-symplectic manifold''. In this case $\mathcal M=T^\star[2]T[1]M$ (see appendix \ref{appendix:gradedsupermanifolds} for the notation), i.e. a degree-shifted version of $T^\star TM$ with its natural graded symplectic structure $\omega$ of degree 2, written in Darboux coordinates as $-\omega=dx^\mu dp_\mu + d\psi^\mu d\chi_\mu$. It is not difficult to see that the most general $\Theta$ that is compatible with $(\Theta,\Theta)=0$ is
\be
\Theta=p_\mu\psi^\mu + \frac{1}{6} H_{\mu\nu\rho}(x)\psi^\mu\psi^\nu\psi^\rho
\ee
with $\partial_{[\mu} H_{\nu\rho\sigma]}=0 \iff dH=0$. Since $(\Theta,\Theta)=0$ implies the Courant algebroid axioms we have therefore recovered the most general exact Courant algebroid. In fact, it was shown in \cite{Roytenberg:2002nu} that \emph{any} Courant algebroid structure on a vector bundle $V$ corresponds to a degree 2 dg-symplectic manifold.

For any dg-symplectic manifold there exist the following two canonically-associated structures: an $L_\infty$-algebra, and a topological field theory. In this case both were deduced a long time ago \cite{roytenberg1998courant,Ikeda:2002wh,Roytenberg:2006qz}. We will only give a revisionist account that anticipates the exceptional story to follow. A mostly self-contained derivation is given in appendices \ref{appendix:gradedsupermanifolds} and \ref{appendix:aksz}.

First, the $L_\infty$-algebra structure: by Proposition \ref{linftyfromdgla} it is defined on the graded vector space $\mathbb L$ formed by all functions of degrees $0$ and $1$ with the new grading assignments $\mathbb L_1=C_1^\infty(\mathcal M), \mathbb L_2=C^\infty(M)$. If we let $A\in \mathbb L_1$, $B \in \mathbb L_2$, the nonzero $L_\infty$ brackets are
\begin{align}
\{B\}=(\Theta,B)\,,\quad \{A_1,A_2\}=\big((\Theta,A_{[1}) A_{2]}\big)\,, \\
\{A,B\}=\frac{1}{2}\big((\Theta,A),B\big)\,,\quad \{A_1,A_2,A_3\}=\Big(\big((\Theta,A_{[1}),A_2\big),A_{3]}\Big)\,.
\end{align}
In the last three we recognise the Courant bracket, the anchor map, and the Jacobiator: from the $k=2$ Jacobi identity \eqref{kjacobirule}
\be
3\big\{\{A_{[1},A_{2}\},A_{3]}\big\}=-\{\{A_{1},A_2,A_3\}\}\propto\psi^\mu\partial_\mu\langle L_{A_{{[1}}} A_2,A_{3]}\rangle\,.
\ee
One can also prove the various Jacobi identies directly from the Courant algebroid axioms as in \cite{roytenberg1998courant,Hohm:2017pnh}.\footnote{Beware however that we are using the convention where all brackets are degree $-1$ and also that their $\mathbb L$ is bigger: in our conventions there is also a nonzero space $\mathbb{L}_3$ housing the constants on $M$. It is unclear what the necessity of that is.} What is interesting here  is that the graded vector space underlying the $L_\infty$-algebra is identical to the tensor hierarchy (functions of degrees 0 and 1 on $\mathcal M$).

Finally, the associated topological field theory. This is given by the AKSZ construction \cite{Alexandrov:1995kv} applied to the dg-symplectic target $\mathcal M$ of degree $p$, which can perhaps be seen as a generalisation of the ``Chern-Simons quantum mechanics'' of \cite{Howe:1989uk}. Making the standard choice for the worldvolume supermanifold, namely $T[1]\Sigma$, produces a topological $p$-brane lagrangian in terms of a Batalin-Vilkovisky master action, as we review in appendix \ref{appendix:aksz}. For closed worldvolumes the master equation is satisfied iff $(\Theta,\Theta)=0$.

Fortunately we will not consider the ghost sector and the prescription for writing down the bosonic lagrangian is easy: promote homogeneous local coordinates $z^a$ on $\mathcal M$ to forms on the $p$-brane worldvolume of rank equal to $\deg z^a$, then write the lagrangian $p$-form
\be
(-1)^{p}\vartheta_adz^a - \Theta(z)
\ee
where $\vartheta_a$ are the components of the canonical symplectic potential 1-form for $\omega$ on $\mathcal M$ \eqref{canonicalsymplecticpotential}. Naturally, starting from the dg-symplectic manifold defining a Courant structure yields the Courant sigma model \cite{Ikeda:2002wh} (a 2-brane since $p=2$). In the exact case the lagrangian is
\be
-p_\mu dx^\mu + \frac{1}{2}(\chi_\mu d\psi^\mu + \psi^\mu d\chi_\mu) -\left(p_\mu \psi^\mu + \frac{1}{6} H_{\mu\nu\rho}(x)\psi^\mu\psi^\nu\psi^\rho\right)\,.
\ee
Here $(x^\mu,\psi^\mu,\chi_\mu,p_\mu)$ are ordinary 0-, 1-, 1-, and 2-forms respectively on the 2-brane worldvolume $\Sigma$. Eliminating $p_\mu,\psi^\mu$ by their equations of motion produces the ordinary electric string WZ coupling to a $B$-field with 3-form field strength $H$, up to a total derivative:
\be
\frac{1}{6} H_{\mu\nu\rho}(x)dx^\mu dx^\nu dx^\rho -\frac{1}{2}dx^\mu d\chi_\mu\,.
\ee

\section{The dg-symplectic geometry of the exceptional tangent bundle}
Consider the problem of expressing the generalised Lie derivative $L_A A'$ of the exceptional tangent bundle in the derived Poisson bracket form
\be
-\big((\Theta,A),A'\big)
\ee
where the Poisson bracket $(\--,\--)$ and hamiltonian $\Theta$ together determine a dg-symplectic structure (see Definition \ref{dgsymplecticmanifold}) on the graded supermanifold $\mathcal M$, all of which are to be determined. Assume that the degree of the symplectic form $\omega$ on $\mathcal M$ is $p$, so the Poisson bracket is of degree $-p$ and $\Theta$ is degree $p+1$, and that the degrees of the functions $A$ and $A'$ representing the corresponding generalised vectors are both $q$.

Unlike the Courant algebroid case we have little to work with in terms of identities that the generalised Lie derivative known to satisfy, \emph{except} identity \eqref{loday}: generalised vectors form a Leibniz algebroid even in the exceptional case. In fact \eqref{loday} always holds in this setup if --- besides $(\Theta,\Theta)=0$ --- we assume that $((\Theta,A),A')$ is of degree $q$ when both $A$ and $A'$ are. This requirement fixes $q=p-1$ and the calculation \eqref{lodayproof1}, \eqref{lodayproof2} goes through unchanged.

Therefore we only need to somehow accommodate the exceptional tangent bundle inside the space of functions of degree $p-1$, and find a suitable hamiltonian $\Theta$. This must be done on a case-by-case basis.

\subsection{M-theory}

Consider the M-theory construction first. The untwisted M-theory exceptional tangent bundle is \cite{Hull:2007zu,Pacheco:2008ps}
\begin{align}
\label{e6tangentmtheoryidentification}
\begin{matrix}
&E&\cong &TM&\oplus &\Lambda^2 T^\star M& \oplus &\Lambda^5 T^\star M\,.\\
&A\in \Gamma(E)&\leftrightarrow &(v\,,& &\omega\,,& &\sigma)
\end{matrix}
\end{align}
where the internal space $M$ is $d$-dimensional. The fibres of $E$ transform in the $R_1$ representations of $E_d\times \mathbb{R}^+ $ as given by table \ref{udualitytable} for $2\leq d \leq 6$. The corresponding generalised Lie derivative (in the absence of twists) is \cite{Pacheco:2008ps}
\begin{equation}
\begin{split}L_{A}A^{\prime} & =\mathcal{L}_{v}v^{\prime}+(\mathcal{L}_{v}\omega^{\prime}-\imath_{v^{\prime}}d\omega)+(\mathcal{L}_{v}\sigma^{\prime}-\imath_{v^{\prime}}d\sigma-\omega^{\prime}d\omega)
\end{split}
\label{eq:M_Dorf_vector}
\end{equation}

 Clearly the $p=6$ graded symplectic manifold $T^\star[6]T[1]M$ accommodates the vector $v$ and 5-form $\sigma$ inside the space of functions at degree $p-1=5$. To get the 2-form as well we add an extra coordinate $\zeta$ of degree $3$. We end up with\footnote{As pointed out in the paper with either the best or worst title \cite{severa2001some}, section 3, principal $\mathbb R[n]$-bundles are trivialisable so there is no gain in generality in considering an $\mathbb{R}[3]$-bundle with base $T^\star[6]T[1]M$.}
 \be
 \mathcal M=T^\star[6]T[1]M\times \mathbb R[3]
 \ee with coordinates
\begin{align}
\begin{matrix}
&(x^\mu\,,& &\psi^\mu\,,& &\zeta\,,& &\chi_\mu \,,& &p_\mu)\\
\deg &0& &1& &3& &5& &6
\end{matrix}
\end{align}
and symplectic structure $\omega=dp_\mu dx^\mu- d\chi_\mu d\psi^\mu-d\zeta d\zeta/2
$ yielding the Poisson brackets
\begin{align}
(x^\mu,p_\nu)&=\delta^\mu_\nu=-(p_\nu,x^\mu)\,,\\
(\psi^\mu,\chi_\nu)&=\delta^\mu_\nu=+(\chi_\nu,\psi^\mu)\\
(\zeta,\zeta)&=1
\end{align}
The functions $A$ at degree $6-1=5$ are expressed as
\be
\label{sectionsasfunctions}
A= v^\mu(x)\chi_\mu + \frac{1}{2!}\omega_{\mu_1\mu_2}(x)\zeta \psi^{\mu_1}\psi^{\mu_2} - \frac{1}{5!}\sigma_{\mu_1\dots \mu_5}(x)\psi^{\mu_1}\cdots\psi^{\mu_5}
\ee
so we only need to find a suitable hamiltonian $\Theta$ of degree $6+1=7$ to complete the construction. The only choice not involving arbitrary functions on $M$ (motivated by the analogous situation in the $O(d,d)$ case) is $\Theta=p_\mu\psi^\mu$. Then
\begin{prop}
\label{genLieMtheory}
With the identification \eqref{sectionsasfunctions} between sections $A,A'$ of $TM\oplus \Lambda^2 T^\star M \oplus \Lambda^5 T^\star M$ and functions $A,A'$ of degree 5, the generalised Lie derivative $L_A A'$ \eqref{eq:M_Dorf_vector} can be written as the following derived Poisson bracket (for $\Theta=p_\mu\psi^\mu$):
\begin{align}
-((\Theta,A),A')=\nonumber\\
\left(v^\mu\partial_\mu v'\,^\nu - v'\,^\mu \partial_\mu v^\nu\right)\chi_\nu\nonumber\\
\left(v^\mu\partial_\mu\omega'_{\nu_1\nu_2} + 2 \partial_{\nu_1} v^\rho \omega'_{\rho\nu_2}-3 v'\,^\rho \partial_{[\rho}\omega_{\nu_1\nu_2]}\right)\frac{1}{2}\zeta\psi^{\nu_1}\psi^{\nu_2}\nonumber\\
\left(v^\mu\partial_\mu \sigma'_{\mu_1\dots\mu_5}+ 5 \partial_{\mu_1}v^\rho\sigma'_{\rho\mu_2\dots\mu_5}-6v'\,^\rho \partial_{[\rho}\sigma_{\mu_1\dots\mu_5]}-\frac{5!}{4}\partial_{\mu_1}\omega_{\mu_2\mu_3}\omega'_{\mu_4\mu_5}\right)\frac{(-1)}{5!} \psi^{\mu_1}\cdots \psi^{\mu_5}\,.
\end{align}
\end{prop}
\begin{proof}
Straightforward calculation. One first finds
\be
(\Theta,A)=p_\mu v^\mu -\partial_\mu v^\nu \psi^\mu \chi_\nu + \frac{1}{5!} \partial_{\mu_1}\sigma_{\mu_2\dots\mu_6} \psi^{\mu_1}\cdots \psi^{\mu_6} + \frac{1}{2} \partial_{\mu_1}\omega_{\mu_2\mu_3} \zeta \psi^{\mu_1}\psi^{\mu_2}\psi^{\mu_3}\,.
\ee
There is a plus sign on the last term from moving a single $\psi$ past $\zeta$. Whence e.g.
\be
((\Theta,A),A')=\dots+\left(\frac{1}{2} \partial_{\mu_1}\omega_{\mu_2\mu_3} \zeta \psi^{\mu_1}\psi^{\mu_2}\psi^{\mu_3}\,, \frac{1}{2} \omega'_{\nu_1\nu_2}\zeta \psi^{\nu_1}\psi^{\nu_2}\right)=-\frac{5!}{4}\partial_{\mu_1}\omega_{\mu_2\mu_3}\omega'_{\mu_4\mu_5}\psi^{\mu_1}\cdots \psi^{\mu_5}
\ee
which correctly reproduces the $-\omega' d\omega$ term in $L_A A'$.
\end{proof}

Thus this is the --- or possibly \emph{a} --- correct dg-symplectic manifold structure for the M-theory $E_d$ generalised Lie derivative, for $2\leq d \leq 6$ and in the absence of twists (which we treat later). The extra odd coordinate $\zeta$ of degree 3 which was invoked to accommodate the 2-form is quite strange at the moment but its significance will be clarified when we discuss the corresponding topological field theory: it will produce the self-dual 2-form on the M5-brane worldvolume.

\subsection{Type IIB}
The untwisted type-IIB exceptional tangent bundle is
\begin{align}
\label{e5tangentiibidentification}
\begin{matrix}
&E&\cong &TM&\oplus &S\otimes T^\star M& \oplus &\Lambda^3 T^\star M\,.\\
&A\in \Gamma(E)&\leftrightarrow &(v\,,& &\lambda^\alpha,& &\rho)
\end{matrix}
\end{align}
where $S$ is an $SL(2)$-doublet bundle (with invariant antisymmetric tensor $\varepsilon_{\alpha\beta}$) over the $(d-1)$-dimensional internal space $M$. The fibres of $E$ transform in the $R_1$ representations of $E_d\times \mathbb{R}^+ $ as given by table \ref{udualitytable} for $2\leq d \leq 5$; the construction doesn't seem to work for $E_6$. The corresponding generalised Lie derivative $L_A A'$ (in the absence of twists) is \cite{Pacheco:2008ps}:
\begin{align}
\label{genliederivativeIIB}
L_A A'= \mathcal L_v v' + (\mathcal L_v \lambda'\,^\alpha -\iota_{v'}d\lambda^\alpha) + (\mathcal L_v \rho' - \iota_{v'} d\rho + \varepsilon_{\alpha\beta} d\lambda^\alpha \lambda'\,^\beta)\,.
\end{align}

Now take
\be
\mathcal M= T^\star[4]T[1]M\times \mathbb{R}^2[2]
\ee
with local homogeneous coordinates
\begin{align}
\begin{matrix}
&(x^\mu\,,& &\psi^\mu\,,& &\zeta_\alpha\,,& &\chi_\mu \,,& &p_\mu)\\
\deg &0& &1& &2& &3& &4
\end{matrix}
\end{align}
and symplectic structure $\omega=dp_\mu dx^\mu- d\chi_\mu d\psi^\mu+d\zeta_\alpha \varepsilon^{\alpha\beta}d\zeta_\beta/2$
 of degree $p=4$ defined by the Poisson brackets
\begin{align}
(x^\mu,p_\nu)&=\delta^\mu_\nu=-(p_\nu,x^\mu)\,,\\
(\psi^\mu,\chi_\nu)&=\delta^\mu_\nu=+(\chi_\nu,\psi^\mu)\\
(\zeta_\alpha,\zeta_\beta)&=\varepsilon_{\alpha\beta}\,.
\end{align}
Note that since $p=4$ is even and the ``extra variables'' $\zeta_\alpha$ are bosonic, they Poisson anticommute in this case so the above formula makes sense.

Sections $A$ of $E$ are identified with functions of degree $p-1=3$:
\be
\label{sectionsofEasfunctionsIIB}
A=v^\mu \chi_\mu +\lambda_\mu^\alpha(x) \psi^\mu \theta_\alpha + \frac{1}{3!}\rho_{\mu\nu\rho}(x) \psi^\mu\psi^\nu\psi^\rho
\ee

It is then easy to verify
\begin{prop}
\label{genLieIIB}
With the identification \eqref{sectionsofEasfunctionsIIB} between sections $A,A'$ of $E$ and functions of degree 3, the generalised Lie derivative $L_A A'$ \eqref{genliederivativeIIB} can be written as the following derived Poisson bracket (for $\Theta=p_\mu\psi^\mu$):
\begin{align}
 - \big((\Theta, A), A'\big)=\nonumber\\
 \left(v^\mu\partial_\mu v'\,^\nu - v'\,^\mu \partial_\mu v^\nu\right)\chi_\nu+  \left(v^\mu\partial_\mu \lambda_\nu'\,^\alpha + \partial_\nu v^\mu\lambda_\mu'\,^\alpha -2v'\,^\mu\partial_{[\mu}\lambda_{\nu]}^\alpha\right)\psi^\nu\zeta_\alpha +\nonumber\\
 \left(v^\sigma\partial_\sigma\rho'_{\mu\nu\rho} + 3 \partial_\mu v^\sigma \rho'_{\sigma\nu\rho}-4v'\,^\sigma\partial_{[\sigma}\rho_{\mu\nu\rho]} +6\varepsilon_{\alpha\beta}\partial_\mu \lambda_\nu^\alpha\lambda_\rho'\,^\beta \right) \frac{1}{3!}\psi^\mu\psi^\nu\psi^\rho\,.
\end{align}
\end{prop}

Again, the role of the extra variables $\zeta_\alpha$ is clarified in the context of the corresponding AKSZ topological field theory, where they reproduce gauge fields on the D3-brane worldvolume, as we see in section \ref{section:branewzterms}.

\subsection{Twists, automorphisms, and Bianchi identities/field equations}
\label{subection:twists}
Like in the $O(d,d)$ case we can recover the known twists of the generalised Lie derivative by considering the most general hamiltonian $\Theta$ consistent with $(\Theta,\Theta)=0$. The one restriction we will impose is that the anchor map $E\to TM$ defined implicitly by the following derived bracket
\be
-\big((\Theta,A),f)\,,\qquad f \in C^\infty(M)\,,\quad A\in C^\infty_{p-1}(\mathcal M)\cong \Gamma(E)
\ee
is \emph{onto} (NB this expression has the correct degree $(p+1)+(p-1)-2p=0$). This is analogous to considering \emph{exact} Courant algebroids in the $O(d,d)$ case.

In the M-theory case degree-counting implies the most general $\Theta$ (here of degree $6+1=7$) takes the form
\be 
\label{thetaexpressionmtheory}
\Theta= a^\mu{}_\nu(x) p_\mu \psi^\nu + b^\mu{}_{\nu\rho}(x)\chi_\mu \psi^\nu\psi^\sigma -F_7(x) \psi^7 + F_4(x) \zeta \psi^4 
\ee
We have used a compact notation where e.g.
\be
F_7(x) \psi^7:=\frac{1}{7!}F_{7\, \mu_1\dots\mu_7} \psi^{\mu_1}\cdots \psi^{\mu_7}\,.
\ee

Now to determine the constraints on $\Theta$ due to $(\Theta,\Theta)=0$.
\emph{Note that by the above assumption $a^\mu{}_\nu$ is invertible}. One can therefore use a symplectomorphism to set $a^\mu{}_\nu=\delta^\mu_\nu$: let $\psi^\mu= g^\mu{}_\nu(x) \psi'\,^\nu\,, \chi_\mu = g^{-1}_\mu{}^\nu(x) \chi'_\nu$. Then
\be
d\psi^\mu d\chi_\mu= d\psi'\,^\mu d\chi'_\mu+ d(g^{-1}_\nu{}^\rho dg^\nu{}_\mu \psi'\,^\mu \chi'_\rho)\,,
\ee
so the symplectic form is invariant if we also shift, schematically, $p_\mu=p'_\mu + g^{-1}\partial g \psi'\chi'$. For $g=a^{-1}$ the effect of this in $\Theta$ is to replace $a^\mu{}_\nu(x) p_\mu \psi^\nu \to p'_\mu \psi'\,^\mu$ at the price of $b^\mu{}_{\nu\rho}(x)$ contributions. However, $(\Theta,\Theta)=0$ for \eqref{thetaexpressionmtheory} for $a^\mu{}_\nu=\delta^\mu_\nu$ implies $b^\mu{}_{\nu\rho}(x)=0$.

Therefore we arrive at
\be 
\label{thetaexpressionmtheory2}
\Theta=  p_\mu \psi^\mu  -F_7(x) \psi^7 + F_4(x) \zeta \psi^4\,.
\ee
An easy calculation then gives,
\begin{align}
\label{mtheoryfieldequation}
\boxed{(\Theta,\Theta)=0 \iff dF_4=0\,,\quad dF_7+\frac{1}{2}F_4F_4=0\,.}&&\text{(M-theory)}
\end{align}
These conditions characterise the known consistent twists of the M-theory generalised Lie derivative \cite{Hull:2007zu} (see e.g. appendix E.1 of \cite{Ashmore:2015joa}). Upon identifying $F_7$ with the Hodge dual of the M-theory 4-form field strength $F_4$, the same conditions can also be interpreted as the Bianchi identity and equation of motion of the 11D supergravity 3-form $C$ (with $dC=F_4$ away from sources).

For the type IIB construction we similarly see that the most general hamiltonian $\Theta$ (now of degree $p+1=5$) that could possibly be consistent with $(\Theta,\Theta)=0$ is
\begin{align}
\label{thetaexpressionIIB}
\Theta=p_\mu \psi^\mu +F_3^\alpha \psi^3\zeta_\alpha + F_5\psi^5 +\frac{1}{2} A_\mu^{\alpha \beta} \zeta_\alpha \zeta_\beta \psi^\mu
\end{align}
where $(F_3^\alpha,F_5, A_\mu^{\alpha \beta})$ are all functions of $x^\mu$. They respectively define a doublet of three-forms, a five-form, and an $sl(2,\mathbb{R})$-valued 1-form (the the bilinears $\zeta_\alpha \zeta_\beta$ form an $sl(2,\mathbb{R})$ subalgebra under the Poisson bracket).

In full, for $a^\mu{}_\nu=\delta^\mu_\nu$ and $v=b=0$ we have
\begin{align}
(\Theta,\Theta)=\left(-\frac{2}{5!} \partial_{\mu_1}F_{5\,\mu_2\dots\mu_6}+\frac{1}{(3!)^2} F^\alpha_{3\,\mu_1\dots\mu_3}\varepsilon_{\alpha \beta} F^\beta_{3\,\mu_4\dots \mu_6}  \right)\psi^{\mu_1} \dots \psi^{\mu_6}+ \nonumber \\
\frac{2}{3!}\left(- \partial_{\mu_1} F^\alpha_{3\,\mu_2\dots\mu_4} + F^\beta_{3\,\mu_1\dots\mu_3}\varepsilon_{\beta \gamma}A^{\gamma \alpha}_{\mu_4}   \right)\zeta_\alpha \psi^{\mu_1} \dots \psi^{\mu_4}+\nonumber\\
\left(-\partial_{\mu} A_\nu^{\alpha \beta}+ A_\mu^{\alpha \gamma}A_\nu^{\delta\beta}\varepsilon_{\gamma\delta}\right)\psi^\mu\psi^\nu\zeta_\alpha\zeta_\beta\,,\label{typeiibthetathetaexplicit}
\end{align}
so
\begin{align}
\label{iibfieldequation}
\boxed{(\Theta,\Theta)=0\iff F_A=0\,,\quad D_AF_3=0\,,\quad d F_5-\frac{1}{2}F_3^\alpha\varepsilon_{\alpha\beta} F_3^\beta=0\,.} &&\text{(Type IIB)}
\end{align}
where $F_A$ is a field strength for the $sl(2,\mathbb{R})$-valued 1-form $A^{\alpha\beta}$ defined explicitly by the final line of \eqref{typeiibthetathetaexplicit}, and where $D_A$ is the associated exterior covariant derivative.

For $A^{\alpha\beta}=0$ we recover the known twists of the type IIB exceptional tangent bundle (see e.g. \cite{Ashmore:2015joa} appendix E.2).

For nonzero $A^{\alpha\beta}$ we have an extra twist by a flat $sl(2,\mathbb{R})$ connection. To interpret $(\Theta,\Theta)=0$ as field equations/Bianchi identities in type IIB supergravity, we first identify $A\propto dg g^{-1}$ where $g$ is an $SL(2,\mathbb{R})/U(1)$ coset representative encoding the IIB axion $C_0$ and dilaton $\phi$, so $F_A=0$ is seen as a Bianchi identity; we also need to relate the 3-form doublet $F^\alpha$ to the RR 3-form $G_3$ and the NS 3-form $H$. We choose
\begin{align}
\label{iibtwistparametrisation}
A^{12}=d\phi/2\,,\;A^{22}=e^\phi d C_0\,, \;(F^1,F^2)=(-e^{-\phi/2}H\,,\:e^{\phi/2}(G_3+C_0H)+e^{\phi/2} C_0 H)\,,
\end{align}
for which the first two equations of \eqref{iibfieldequation} are satisfied due to the IIB Bianchi identities $dH=d(G_3+C_0 H)=0$. For $F^\alpha$ this identification is ambiguous up to a constant matrix multiplying $(G_3+C_0 H,H)$. Then the last equation of \eqref{iibfieldequation} reads
\be
dF_5+ H G_3=0
\ee
which is the IIB RR 5-form Bianchi identity if we set $F_5=-G_5$.

The upshot is that in both M-theory and type IIB constructions,
\begin{center}
\emph{$(\Theta,\Theta)=0$ is equivalent to Bianchi identities/field equations for the fields specifying the twists.}
\end{center}
These should probably be seen as Bianchi identities possibly involving dual potentials, however.

\bigskip
Checking that the twists as we described above enter the generalised Lie derivative correctly is a trivial exercise in calculating $-((\Theta,A),A')$ for $\Theta$ as above and matching against the expressions given in e.g. \cite{Ashmore:2015joa}. It is better to derive them from automorphisms of the graded Poisson structure  on $\mathcal M$. Insofar as we regard the twists as characterising (part of) a supergravity background,
\begin{center} \emph{automorphisms should be seen as \emph{gauge transformations} relating equivalent backgrounds (i.e. twists)}
\end{center} (Note: this is \emph{not} the same notion as that of automorphisms of the generalised Lie derivative; those must also preserve $\Theta$.)

As pointed out in appendix \ref{automorphismsgradedPoisson}, an infinitesimal automorphism $X$ is always inner; there exists a function $R$ of degree $p$ so $X=X_R=(R,\--)$ ($X$ must be of degree zero). Therefore
\begin{itemize}
	\item In the M-theory case infinitesimal automorphisms are generated by $R\in C_6^\infty(\mathcal M)$:
	\be
	R=u^\mu(x)p_\mu+r^\mu{}_\nu(x)\psi^\nu\chi_\mu + a_3(x)\zeta \psi^3 + \tilde a_6(x) \psi^6\,.
	\ee
The last three terms generate a subgroup of the local $E_d\times \mathbb{R}^+$ action on $E\cong C_5^\infty(\mathcal M)$ (c.f. formula (E.6) of \cite{Ashmore:2015joa} for $l,\alpha,\tilde \alpha=0$) of $GL(d)$ transformations along with shifts by 3-forms $a$ and 6-forms $\tilde a$. Indeed, on a generalised vector $A\in C_5^\infty(\mathcal M)$
\begin{align}
X_R\cdot A=(R,A)=(-r^\mu{}_\nu v^\nu)\chi_\mu +\left(2 r^\mu{}_{\nu_1} \omega_{\mu\nu_2}+ v^\mu a_{\mu\nu_1\nu_2} \right) \frac{1}{2}\zeta \psi^{\nu_1}\psi^{\nu_2}\nonumber\\
+\left(5 r^\mu{}_{\nu_1} \sigma_{\mu \nu_2\dots\nu_5}+2 a_{\nu_1\nu_2\nu_3}\omega_{\nu_4\nu_5}+v^\mu \tilde a_{\mu\nu_1\dots\nu_5}  \right)\frac{-1}{5!}\psi^{\nu_1}\dots\psi^{\nu_5}
\end{align}
Furthermore for $\Theta=p_\mu\psi^\mu$ and $u,r=0$ we calculate
\begin{align}
X_R \cdot\Theta=-\frac{1}{3!}\partial_{\mu_1}a_{\mu_2\cdots \mu_4} \zeta \psi^{\mu_1}\cdots \psi^{\mu_4} + \frac{1}{6!}\partial_{\mu_1}\tilde a_{\mu_2\cdots \mu_7} \psi^{\mu_1}\cdots \psi^{\mu_7}\,,\\
X_R^2\cdot\Theta=+\frac{1}{3!3!}a_{\mu_1\cdots\mu_3}\partial_{\mu_4}a_{\mu_5\cdots\mu_7}\psi^{\mu_1}\cdots\psi^{\mu_7}\,, \qquad X_R^3\cdot \Theta=0\,.
\end{align}
Hence automorphisms of the generalised Lie derivative ($X_R\cdot \Theta=0$) consist of shifts by \emph{closed} 3- and 6-forms  besides the action of $GL(d)$; $R$ such that $X_R\cdot \Theta\neq 0$ on the other hand introduces 4- and 7-form twists $-F_7=d\tilde a+ada/2,-F_4=da$. Given $X_R$ acts on sections of $E$ by infinitesimal $E_d\times \mathbb{R}^+$ transformations this proves that $F_4,F_7$ enter the generalised Lie derivative correctly.

Quotienting $F_4,F_7$ by ``trivial'' twists $\exp(X_R)$ (for $R$ with $u,r=0$ with globally-defined 3-form $a$ and 6-form $\tilde a$) leads to the ``non-abelian de Rham cohomology'' characterising equivalent exceptional Leibniz algebroid structures as  was previously found in \cite{Baraglia:2011dg}, and should be thought of as the analogous result to the classification of exact Courant algebroids by $H^3(M)$ \cite{vsevera2015poisson}.

\item Similarly for type IIB $R$ must be degree 4, giving
\be
R=u^\mu(x)p_\mu+r^\mu{}_\nu(x)\psi^\nu\chi_\mu+ a^{\alpha\beta}(x)\zeta_\alpha \zeta_\beta+ B_2^\alpha(x) \zeta_\alpha \psi^2 + C_4(x) \psi^4\,.
\ee
The last four terms again generate  a subgroup of $E_{d}\times \mathbb{R}^+$acting on $E\cong C_3^\infty(\mathcal M)$, see appendix E.2 of \cite{Ashmore:2015joa}, specifically (E.35). This can be analysed like the previous case.

\end{itemize}

\section{$L_\infty$ and the tensor hierarchy}
As we review in appendix \ref{appendix:linfinity}, for any dg-symplectic manifold $\mathcal M$ there is a canonical $L_\infty$-algebra $\mathbb L$: the Poisson algebra $C^\infty(\mathcal M)$ is a graded Lie algebra (after a grading change) and for any graded Lie algebra the results of \cite{getzler2010higher,fiorenza2007structures} show there is an associated $L_\infty$ algebra. As a graded vector space
\be
\mathbb L=C_{p-1}^\infty(\mathcal M)\oplus C_{p-2}^\infty(\mathcal M)\dots \oplus C_{0}^\infty(\mathcal M)\oplus \dots
\ee
which terminates at $C_{0}^\infty(\mathcal M)$ if $\mathcal M$ has no coordinates of negative degree, as is the case in this paper. The grading is changed so $C_{p-n}^\infty(\mathcal M)$ is degree n in $\mathbb L$; we will accordingly write $\mathbb L_n:=C_{p-n}^\infty(\mathcal M)$.

That is the construction we use to associate $L_\infty$-algebras to the exceptional generalised Lie derivative in both M-theory and type IIB constructions in this paper. In the M-theory case an $L_\infty$-algebra structure (but not the $L_\infty$-algebroid) was first found in \cite{Baraglia:2011dg}, and the EFT generalisation thereof in \cite{Cederwall:2018aab}. The type IIB construction seems original.

What is interesting from a physics standpoint is that for both $E_d$ constructions considered in this paper and for $2\leq d \leq 6$,
\begin{align}
{\mathbb L}_1\,,{\mathbb L}_2\,,\dots \leftrightarrow R_1\,,R_2\,,\dots
\end{align}
where $R_n$ on the right-hand side are the modules of table \ref{udualitytable} characterising the EFT tensor hierarchy! The identification is valid up to a certain $R_n$ (depending on $d$ and the construction considered). When the representations do match the identification is precise: the fields on the left-hand side have the correct weight \eqref{godweight} under the generalised Lie derivative.

Notwithstanding the weight issues which we deal with shortly, the best way to confirm that the modules $\mathbb L_1,\mathbb L_2,\dots$ transform as the corresponding modules $R_1,R_2,\dots$ in the tensor hierarchy whenever they agree is to mimic the derivation of the tensor hierarchy in \cite{Hohm:2015xna}. First write $A\in\mathbb L_1, B\in \mathbb L_2,\dots$ as for $R_1,R_2,\dots$ that paper. To check that $X_{-(\Theta,A)}$ generates the Dorfman derivative with respect to $A$ in $\mathbb L_2$ (for $\Theta=p_\mu\psi^\mu$; we assume no twists, to match the EFT literature), notice that $(A_1,A_2)\in \mathbb L_2\, \forall A_1,A_2\in\mathbb L_1$. Then since $X_{-(\Theta,A)}$ is of degree zero,
\be
X_{-(\Theta,A)}\cdot (A_1,A_2)= (X_{-(\Theta,A)}\cdot A_1,A_2)+ (A_1,X_{-(\Theta,A)}A_2)\,.
\ee
By comparing with the argument of \cite{Hohm:2015xna} from (2.32) to (2.35) we see the claim follows \emph{if we can identify the Poisson bracket $(A_1,A_2)$ with the bullet product $\bullet$ as defined in that reference} (i.e. the well-known symmetric map $R_1\otimes R_1 \to R_2$), given that the generalised Lie derivative of generalised vectors in $R_1$ is correctly reproduced by $X_{-(\Theta,A)}$. This must be done separately for the M-theory and type IIB constructions.

More generally, if all the relevant Poisson brackets match the corresponding bullet products, $X_{-(\Theta,A)}$ will correctly generate the generalised Lie derivative on the corresponding module $\mathbb L_n$. Given that, if we use identity \eqref{pbhomomorphism} (expressing the homomorphism between the Lie derivative of graded vector fields and Poisson bracket) to find
\be
\label{magicproof}
-X_{(\Theta,A)}\cdot T=-(X_\Theta X_A + X_A X_\Theta)\cdot T\qquad \forall T\in\mathbb L_n
\ee
\emph{we arrive at the magic formula for the generalised Lie derivative} \cite{Wang:2015hca}: $X_A$ acts by bullet product on all modules $\mathbb L_n$ as previously established, thus \eqref{magicproof} is only consistent with the magic formula
\be
L_A T=\hat \partial (A\bullet T) + A\bullet \hat\partial T.
\ee
if we identify $X_\Theta \sim \hat \partial$. This identification makes sense because $X_\Theta$ is a map $\mathbb L_n\to\mathbb L_{n-1}$ ($X_\Theta$ is degree 1 as a derivation of $C^\infty(\mathcal M)$ and we identified $\mathbb L_n\cong C^\infty_{p-n}(\mathcal M)$).

\subsection{A weighty topic}

Figuring out the weights requires matching conventions between EGG and EFT. First let $V^M$ be an EFT generalised vector of weight $\lambda_V$. Under the EFT gen. Lie derivative with parameter $A^N$ we have ($\omega$ as in \eqref{godweight}, $Y$ is as in \cite{Berman:2012vc})
\be
\delta_A V^M=  A^P \partial_P V^M- V^P \partial_P  A^M + Y^{MN}{}_{PQ}\partial_N  A^P V^Q + ( \lambda_V +\omega) \partial_P  A^P V^M\,.
\ee
Fix a solution of the section condition ($Y^{MN}{}_{PQ}\partial_M\partial_N=0$) so $\partial_M=(\partial_\mu,\partial_A)$ with $\partial_A=0$, and the section $M$ has local coordinates $x^\mu$. Infinitesimal diffeomorphisms acting on the section arise from the vector component of $ A^M$, so set $ A^M= \lambda^\mu \delta_\mu^M$. Set also $V^M=\delta_\mu^M v^\mu$ to check whether the vector component of $V$ is a density or not and at what $\lambda_V$. Since $Y^{\mu\nu}{}_{MN}=0 \implies Y^{MN}{}_{\mu\nu}=0$ (at least for $E_{d\leq 6}$, see e.g. \cite{Arvanitakis:2018hfn}) we get
\be
\delta_\lambda v^\mu=\lambda^\nu\partial_\nu v^\mu - v^\nu\partial_\nu \lambda^\mu +0 +(\lambda_V +\omega) \partial_\nu \lambda^\nu V^\mu
\ee
so we can unambiguously say $v^\mu$ is a genuine vector field on the section --- as opposed to a vector density --- whenever $\lambda_V+\omega = 0$. In EGG on the other hand for an $E_d\times\mathbb R^+$ generalised vector (in the notation of e.g. \cite{Ashmore:2015joa})
\be
\label{weight4}
L_ A V^M=  A^P \partial_P V^M- (\partial\times_\text{adj} A)^M{}_P V^P
\ee
which is explicitly written in terms of tensors on the section $M$ as in Propositions \eqref{genLieMtheory}, \eqref{genLieIIB} here. From the expressions therein we see that the transformation \eqref{weight4} includes no term $\partial_i\lambda^i$, hence that the vector component of $V$ in \eqref{weight4} is also a genuine tensor.

The conclusion is that an EGG generalised vector $V^M$ as is usually written corresponds to an EFT generalised vector of weight $\lambda_V=-\omega$, which transforms the same as a function $V\in C_{p-1}^\infty(\mathcal M)$ in our dg-symplectic manifold construction of the generalised Lie derivative. Therefore the component $t^{i_1\dots i_n}$ of an EFT generalised tensor $T^{M_1 M_2\dots M_n}$ (all indices up) of weight $-n\omega$ (under the \emph{generalised} Lie derivative) will also transform as a genuine tensor and not as a tensor density. Since the fields in the tensor hierarchy for all $R_n (n>1)$ take that form \cite{Aldazabal:2013via} and carry that weight this suggests that fields in the tensor hierarchy are genuine tensors as far as ordinary diffeomorphisms of the section $M$ are concerned. At the same time, if $T \in C_n^\infty(\mathcal M) \,, (0\leq n \leq (p-1))$ and the generalised Lie derivative $L_A T$ is defined via derived Poisson bracket as before,
\be
-\big((\Theta,A),T)
\ee
for $A=\lambda^\mu(x) \chi_\mu$ we will always get the ordinary Lie derivative with respect to $\lambda^\mu$ with no density terms. This strongly suggests that each level of the $L_\infty$ algebra constructed above transforms under generalised Lie derivative exactly like the corresponding module in the tensor hierarchy.

A caveat in that conclusion is that one could dualise (some of) the tensors inside $T$ using the invariant antisymmetric symbols $\varepsilon_{\mu_1\mu_2\dots \mu_{\dim M}},\varepsilon^{\mu_1\mu_2\dots \mu_{\dim M}}$, thereby yielding tensor densities. In fact, compared to the presentation in $C_n^\infty(\mathcal M)$, some of the tensors (on $M$) inside a generalised tensor $T$ in the EFT tensor hierarchy are almost always dualised that way. This seems to preclude a systematic comparison dealing with all cases at once and for this reason we will only consider only the cases $E_2\cong SL(2)\times \mathbb{R}^+\,, E_5\cong {\rm Spin}(5,5)\,, E_6$.

\subsection{M-theory hierarchy}
In the M-theory case ($\mathcal M=T^\star[6]T[1]M\times \mathbb{R}[3]$, $\dim M=d$, $2\leq d\leq 6$),
\begin{align*}
\boxed{\text{functions in $C^\infty(\mathcal M)$ of degrees 5, 4, \dots $d-1$ form $E_d$ representations $R_1$, $R_2$, \dots $R_{7-d}$\,.}}
\end{align*}
In other words the associated $L_6$ algebra spans the tensor hierarchy for the corresponding $E_d$ EFT excepting just the final rep $R_{8-d}\cong \bar R_1$; $\mathbb L_{8-d}$ is smaller than $R_{8-d}$ (cf. table \ref{udualitytable}). The boxed statement follows by comparing with the EFT literature after tabulating the $GL(d)$ representation content of the functions at each degree:
\begin{align*}
\begin{array}{r|cccccc}
\deg &0 &1 &2 &3 &4 &5 \\
\hline
GL(d) &\mathbb{R} &T^\star M &\Lambda^2 T^\star M &\mathbb{R}\oplus \Lambda^3T^\star M &T^\star M\oplus \Lambda^4T^\star M  & T M \oplus \Lambda^2 T^\star M \oplus \Lambda^5 T^\star M\\
E_d\: & R_6 &R_5 &R_4 &R_3 &R_2 &R_1 
\end{array}
\end{align*}

We now explicitly show the relation to the tensor hierarchy for $E_6$, $E_5\cong{\rm Spin}(5,5)$ and $E_2\cong SL(2)\times\mathbb{R}^+$, including the correct transformation under the generalised Lie derivative. For $E_6$ this is in fact trivial because the only module $\mathbb L_n$ that matches the corresponding $R_n$ is $\mathbb L_1$ for which the claim is true by Proposition \ref{genLieMtheory}. For ${\rm Spin}(5,5)$ we also need to check $\mathbb L_2$. $R_2$ is the vector rep of ${\rm Spin}(5,5)$ branching as $\bf 10\to \bf 5 \oplus \bar{\bf 5}$ under the $SL(5)$ subgroup (see e.g. \cite{Wang:2015hca}), matching $\mathbb L_2\cong T^\star M \oplus \Lambda^4 T^\star M$. We also need to check that the Poisson bracket reproduces the bullet product. Write \eqref{sectionsasfunctions} for $A \in \mathbb L_1$ then calculate
\be
-(A,A)/2= (\iota_v \omega)_1 \zeta \psi^1 + (\iota_v \sigma -\omega\wedge\omega/2)\psi^4\qquad A\in\mathbb L_1\,.
\ee
This indeed matches $\bullet: R_1\otimes R_1\to R_2$ (which is symmetric) as given in a $GL(d)$-manifest notation in \cite{Coimbra:2011ky} formula (C.15)\footnote{To see that is indeed the same as the corresponding EFT bullet product notice (2.27) in that paper relating their bundle $N$ to the section condition.}.

In the case of $E_2\cong SL(2)\times \mathbb R^+$ which has the longest tensor hierarchy, $d=\dim M=2$ so the spaces $\mathbb L_n$ simplify considerably. To match the $SL(2)\times \mathbb R^+$ EFT \cite{Berman:2015rcc} we need to dualise using the $SL(2)$ invariant tensor densities $\varepsilon_{\mu\nu}$ and $\varepsilon^{\mu\nu}$ on $M$. We display the result in the notation of $SL(2)\times \mathbb R^+$ appendix A (with the replacement $\alpha\to \mu$):
\begin{flalign*}
\begin{array}{c|c|c|c|c|c|}
F_s &E^{\mu,\nu\rho,ss}(x)\varepsilon_{\mu\sigma}\varepsilon_{\nu\rho}\psi^\sigma &D^{\mu\nu,ss}(x)\varepsilon_{\mu\nu}\varepsilon_{\rho\sigma}\psi^\rho\psi^\sigma &C^{\mu\nu,s}(x)\varepsilon_{\mu\nu}\zeta&B^{\mu,s}(x)\varepsilon_{\mu\nu}\psi^\nu\zeta  & A^\mu(x)\chi_\mu +A^s(x)\varepsilon_{\mu\nu}\zeta \psi^\mu\psi^\nu\\\hline
\supset R_6 &R_5 &R_4 &R_3 &R_2 &R_1 
\end{array}
\end{flalign*}
As previously announced the correspondence breaks at $R_{8-d}=R_6$: $R_6$ is reducible and we only get the $SL(2)$ singlet. One verifies by inspection that all Poisson brackets involving only objects in $\mathbb L_5$ or below match the bullet products ((A.3) of \cite{Berman:2015rcc}) involving $R_5$ and below or otherwise of the form $(A,\text{any})$, proving that $X_{-(\Theta,A)}$ generates the generalised Lie derivative (up to constant relative coefficients which we dropped above). As a check, consider $\hat \partial$ as defined in (A.4) of \cite{Berman:2015rcc} and given here in an M-theory section:
\begin{align}
 \left( \hpartial \Ab \right)^{\mu} &= 0 \,, & \left( \hpartial \Ab \right)^{s} &= \partial_\mu \Ab^{\mu,s} \,, & \left( \hpartial \Ac \right)^{\mu,s} &= \partial_ \nu \Ac^{[ \nu\mu],s} \,, \nonumber \\
 \left( \hpartial \Ad \right)^{[\mu \nu],s} &= 0 \,, & \left( \hpartial \Ae \right)^{[\mu \nu],ss} &= \partial_ \rho \Ae^{ \rho,[\mu \nu],ss} \,, & \left( \hpartial \Af \right)^{ \rho,[\mu \nu],ss} &= 0\,.
\end{align}
Since $X_\Theta:=(p_\mu\psi^\mu,\--)=-\psi^\mu\partial_\mu$ on $\mathbb L$ it indeed agrees with $\hat \partial$ if the dualisations are understood properly. Example:
\be
X_\Theta\cdot(E^{\mu,\nu\rho,ss}\varepsilon_{\mu\sigma}\varepsilon_{\nu\rho}\psi^\sigma)=\partial_\tau E^{\mu,\nu\rho,ss}\varepsilon_{\mu\sigma}\varepsilon_{\nu\rho}\psi^\sigma\psi^\tau=\frac{1}{2}\partial_\mu E^{\mu,\nu\rho,ss}\varepsilon_{\nu\rho}\psi^\sigma\psi^\tau \varepsilon_{\sigma\tau }\,.
\ee

\subsection{Type IIB hierarchy}
In the type IIB case ($\mathcal M=T^\star[4] T[1]M \times\mathbb{R}^2[2],\, \dim M=d,\, 1\leq d \leq 4$)
\begin{align*}
\boxed{\text{functions in $C^\infty(\mathcal M)$ of degrees 3, 2, \dots $d-1$ form $E_{d+1}$ representations $R_1$, $R_2$, \dots $R_{5-d}$\,.}}
\end{align*}
In other words the associated $L_4$ algebra spans the tensor hierarchy for the corresponding $E_d$ EFT excepting the last two reps $R_{6-d}$ and $R_{7-d}\cong \bar R_1$; $\mathbb L_{6-d}$ and $\mathbb L_{7-d}$ are smaller than required (cf. table \ref{udualitytable}). The boxed statement follows by comparing with the EFT literature after tabulating the $GL(d)$ representation content of the functions at each degree:
\begin{align*}
\begin{array}{r|cccc}
\deg &0 &1 &2 &3  \\
\hline
GL(d)\times SL(2) &\mathbb{R} &T^\star M &\Lambda^2 T^\star M\oplus S &TM\oplus (S\otimes T^\star M)\oplus \Lambda^3 T^\star M\\
E_{d+1}\:  &R_4 &R_3 &R_2 &R_1 
\end{array}
\end{align*}

Like in the M-theory case, the claim follows for $d=4 \iff E_5\cong {\rm Spin}(5,5)$, by Proposition \ref{genLieIIB}. We also detail the correspondence in the case $d=1 \iff E_2\cong SL(2)\times \mathbb{R}^+$ \cite{Berman:2015rcc} (where now $x^s$ is the single local coordinate on $M$, and $\alpha,\beta$ are $SL(2)$ indices of the bundle $S$):
\begin{align*}
\begin{array}{c|c|c|c}
D^{\alpha\beta,ss}(x)\varepsilon_{\alpha\beta} &C^{\alpha\beta,s}(x)\varepsilon_{\alpha\beta}\psi^s&B^{\alpha,s}(x)\zeta_\alpha & A^s(x)\chi_s +A^\alpha(x)\zeta_\alpha \psi^s\\\hline
R_4 &R_3 &R_2 &R_1 
\end{array}
\end{align*}
The nilpotent derivative $\hat \partial$ on a type IIB section is
\begin{align}
 \left( \hpartial \Ab \right)^{\alpha} &= \partial_s \Ab^{\alpha,s} \,, & \left( \hpartial \Ab \right)^{s} &= 0\,, & \left( \hpartial \Ac \right)^{\alpha,s} &= 0 \,, \nonumber \\
 \left( \hpartial \Ad \right)^{[\alpha\beta],s} &= \partial_s \Ad^{[\alpha\beta],ss} \,, & \left( \hpartial \Ae \right)^{[\alpha\beta],ss} &= 0 \,, & \left( \hpartial \Af \right)^{\gamma,[\alpha\beta],ss} &= \epsilon^{\gamma\delta} \partial_s F_\delta 
\end{align}
which clearly agrees with the action of $X_\Theta=(p_s\psi^s,\--)=-\psi^s\partial_s$ (again, the match to the EFT literature is in the absence of twists).

\section{The topological field theories and M5/D3 Wess-Zumino terms}
\label{section:branewzterms}
To any dg-symplectic manifold $\mathcal M$ the construction \cite{Alexandrov:1995kv} of Alexandrov, Schwarz, Zaboronsky and Kontsevich (AKSZ) (which we review in appendix \ref{appendix:aksz}) associates a topological field theory given in terms of a (classical) Batalin-Vilkovisky master action \eqref{AKSZ}. If the symplectic form $\omega$ is of degree $p$, then the usual construction produces a topological $p$-brane lagrangian propagating in $\mathcal M$. Schematically, the lagrangian $(p+1)$-form is
\be
\vartheta- \Theta
\ee
where $\vartheta$ is the \emph{canonical} symplectic potential \eqref{canonicalsymplecticpotential} for $\omega=-d\vartheta$.

Since our M-theory and type IIB constructions involve $\omega$ of degrees 6 and 4 respectively, we get topological 6- and 4-branes. These depend on twists through $\Theta$ as in \eqref{thetaexpressionmtheory2} and \eqref{thetaexpressionIIB} respectively, whence the functionals
\begin{align}
\label{mtheorybosonicaksz}
S_{\rm M}=\int_{\Sigma_7}  -p_\mu dx^\mu+\frac{1}{6}(\psi^\mu d\chi_\mu   + 5  \chi_\mu d\psi^\mu)+\frac{1}{2}\zeta d \zeta  -\left(p_\mu \psi^\mu  -F_7(x) \psi^7 + F_4(x) \zeta \psi^4\right)\\ 
\label{iibbosonicaksz}
S_{\rm IIB}=\int_{\Sigma_5} -p_\mu dx^\mu +\frac{1}{4}( \psi^\mu d\chi_\mu + 3 \chi_\mu d\psi^\mu)+\frac{1}{2}\varepsilon^{\alpha\beta}\zeta_\alpha d\zeta_\beta\nonumber\\ -\left(p_\mu \psi^\mu +F_3^\alpha(x) \psi^3\zeta_\alpha + F_5(x)\psi^5 +\frac{1}{2} A_\mu^{\alpha \beta}(x) \zeta_\alpha \zeta_\beta \psi^\mu\right)\,.
\end{align}
In the M-theory action $x^\mu,\psi^\mu,\zeta,\chi_\mu,p_\mu$ are respectively 0-, 1-, 3-, 5-, and 6-forms on $\Sigma_7$, etc. for the type IIB action. (We have dropped the (anti)ghosts.) The equations of motion are equivalent to
\be
dz^a=(\Theta,z^a)\qquad (z^a=(x^\mu,p_\mu,\dots))\,.
\ee
We assume $\Sigma_7\,,\:\Sigma_5$ have boundaries $W_6\,, \: W_4$ respectively.

We relate these functionals to the Wess-Zumino terms describing the coupling of probe M5- and D3-branes (of worldvolumes $W_6\,,\: W_4$) to the form fields of 11D and IIB supergravity (as described by the twists in section \ref{subection:twists}) by imposing \emph{consistent boundary conditions} for the variational problem, in the sense that solutions of the equations of motion are actual stationary points of $S_{\rm M},S_{\rm IIB}$. The necessary and sufficient condition is that the respective boundary terms
\begin{align}
\label{m5boundaryterm}
\int_{W_6}-p_\mu \delta x^\mu - \frac{1}{6}(\psi^\mu \delta\chi_\mu+5\chi_\mu \delta \psi^\mu) - \frac{1}{2}\zeta \delta\zeta\,,   \\
\label{d3boundaryterm}
 \int_{W_4}-p_\mu \delta x^\mu - \frac{1}{4}(\psi^\mu \delta\chi_\mu+3\chi_\mu \delta \psi^\mu)+ \frac{1}{2}\varepsilon^{\alpha\beta}\zeta_\alpha \delta\zeta_\beta
\end{align}
vanish on-shell.

In both cases the $\psi^\mu$ and $p_\mu$ equations of motion are $p_\mu=d\chi_\mu+(\text{twist contribs.})\,,\:\psi^\mu=-dx^\mu$, so we can only impose a boundary condition on $\chi_\mu,\zeta,\zeta_\alpha$. (We do not impose a boundary condition on $x^\mu$ because that is unnatural from the point of view of a theory on $W_p$.) Of those, $\chi_\mu$ is a lagrange multiplier for $d\psi^\mu=0$ which only enters the equations of motion as just shown; \emph{any} value of $\chi_\mu$ is consistent with the equations of motion.

\subsection{M5}
\label{section:m5}
First use the $p_\mu,\psi^\mu$ equations of motion to rewrite the boundary term
\be
\int_{W_6}\delta(-\chi_\mu dx^\mu/6)-\iota_{\delta x}F_7 + \zeta \iota_{\delta x}F_4 - \frac{1}{2}\zeta \delta \zeta
\ee
from which we see that $\chi_\mu$ is completely irrelevant for consistency of the variational principle in the presence of a boundary so we simply set it to zero. Then using the identities $\iota_{\delta x}F_7=\iota_{\delta x} dC_6-\iota_{\delta x}(C_3F_4)/2$ (for locally-defined potentials $dC_3=F_4, dC_6= F_7+C_3F_4/2$), $\iota_{\delta x} dC_6=\delta C_6-d\iota_{\delta x} C_6$ and $\iota_{\delta x }(C_3F_4)= C_3 \delta C_3 - 2 C_3 \iota_{\delta x} F_4 + d(C_3 \iota_{\delta x} C_3)$ we find the boundary integrand
\be
d(\iota_{\delta x} C_6+C_3\iota_{\delta x}C_3/2)-\delta C_6 + (\zeta-C_3) \iota_{\delta x}F_4 + \frac{1}{2}(C_3\delta C_3-\zeta \delta \zeta)\,.
\ee

Now consider the $\zeta$ equation of motion
\begin{align}
\label{zetaeom}
d\zeta=\iota^\ast F_4
\end{align}
where we now displayed explicitly the pullback by the embedding $\iota$ of the brane on spacetime described by the $x^\mu$. (As an aside we point out that this equation together with $dF_4=0$ on $M$ states that $(F_4,\zeta)$ represent a cohomology class on $M$ relative to the brane worldvolume $\Sigma_7$. A fortiori this is also a class relative to $W_6$, in agreement with a proposal by Kalkkinen and Stelle \cite{Kalkkinen:2002tk}). The general solution to this on $W_6$ is $\zeta=C_3 + h$ for $h$ a \emph{closed} 3-form and we eventually find the boundary integrand
\be
-\delta \left(C_6-\frac{1}{2} h C_3\right)+d(\iota_{\delta x} C_6+C_3\iota_{\delta x}C_3/2-h \iota_{\delta x} C_3)-\frac{1}{2}h\delta h\,.
\ee
Upon discarding the total derivative\footnote{If $F_7$ is a nontrivial class this is dangerous because $\chi_\mu$ involves the potential $C_6$ for it. The danger and a potential resolution is explained in e.g. \cite{Keir:2013jga}. For terms involving $C_3$ this seems safer due to \eqref{zetaeom}.} we conclude we need to impose a boundary condition on $h$ that makes the last term vanish.

Equivalently, \emph{the necessary and sufficient condition for consistency of the variational principle (after we subtract off the total variations) is that $h$ lie in some isotropic subspace inside the space of 3-forms on $W_6$ with the natural symplectic form $\Omega$ ($\forall \alpha,\beta \in \Lambda^3 T^\star W_6\,, \;\Omega(\alpha,\beta):=\int_{W_6}\alpha \beta$)}. In the absence of any other input we might as well let $h$ lie in a maximal isotropic i.e. lagrangian subspace. There is no natural choice without assuming additional structure, so we invoke a Lorentzian metric on $W_6$ and impose a self-duality condition
\be
\star_6h=\pm h\,.
\ee

Putting everything together,
\begin{align}
\delta\left(S_{\rm M}+\int_{W_6}\left[C_6-\frac{1}{2}h C_3\right]\right)=0 &&\text{on-shell}\nonumber\\
\chi_\mu|_{W_6}=0\,, \:\zeta|_{W_6}=C_3+h\,,\qquad \star_6 h= \pm h\,, \: dh=0\,.
\end{align}
Comparing with the M5-brane Wess-Zumino term \cite{Aharony:1996wp,Pasti:1997gx}, we can express this as
\begin{align}
\boxed{-S_{\rm M}=S_\text{M5, WZ}=\int_{W_6}C_6-\frac{1}{2}h C_3}  &&\text{on-shell.}
\end{align}
What is striking here is how \emph{we obtained the field strength of the chiral 2-form on the M5 worldvolume with essentially no input.} (In fact the qualifier ``essentially'' is possibly superfluous: since the necessary and sufficient condition is that $h$ lie in an isotropic subspace, take a lagrangian subspace containing it and try to find a Lorentzian metric such that the lagrangian consists of self-dual forms for that metric.)

It is therefore tempting to think of $S_{\rm M}$ as the correct form of the M5-brane Wess-Zumino coupling not only when the potentials exist so the above manipulations make sense but also in topologically nontrivial situations. A near-identical proposal in this context has in fact been already made by Kalkkinen and Stelle \cite{Kalkkinen:2002tk}: assume $\Sigma_7$ is now closed and is the boundary of a $V_8$. Then, using \eqref{mtheoryfieldequation} and the $p_\mu,\psi^\mu,\zeta$ equations, $S_{\rm M}$ reads
\be
\label{stellem5wz}
S_{\rm M}=\int_{V_8} \frac{1}{2}F_4F_4 -\int_{\Sigma_7}\frac{1}{2}\zeta F_4=\frac{1}{2}\int_{(V_8,\Sigma_7)}(F_4F_4,\zeta F_4)\,.
\ee
This is the fivebrane Wess-Zumino term in (31) of \cite{Kalkkinen:2002tk} for vanishing gravitational correction terms. (Our $\zeta $ is their $h$: by \eqref{zetaeom} and \eqref{mtheoryfieldequation}, $(F_4,\zeta)$ represents a relative cohomology class.) We refer to that paper for the relation to the ``Hopf-Wess-Zumino'' term of Intriligator \cite{Intriligator:2000eq}.

\subsection{D3}
Take the boundary term \eqref{d3boundaryterm} and use the $p_\mu\,,\psi^\mu$ equations of motion to cast it into the form
\be
\int_{W_4}-\frac{1}{4}\delta(\chi_\mu dx^\mu)+\frac{1}{2}\iota_{\delta x} A^{\alpha\beta}\zeta_\alpha \zeta_\beta + \iota_{\delta x} F_3^\alpha \zeta_\alpha + \iota_{\delta x} F_5 + \frac{1}{2} \varepsilon^{\alpha\beta}\zeta_\alpha \delta \zeta_\beta\,.
\ee
We are again free to set $\chi_\mu=0$ consistently with the variational problem. We introduce $(C_4,C_2^\alpha,h_\alpha)$, such that equations \eqref{iibfieldequation} and the $\zeta_\alpha$ field equation
\be
\label{zetaalphafieldequation}
d\zeta_\alpha=-\varepsilon_{\beta\alpha}(F_3^\beta+A^{\beta\gamma}\zeta_\gamma)
\ee
are satisfied on $M$ and $W_4$ respectively
\begin{align}
F^\alpha_{3\,\mu\nu\rho}=dC^\alpha_{2} - A^{\alpha\beta} C_{2}^\gamma \varepsilon_{\beta\gamma}\,,\qquad F_5=dC_4+\frac{1}{2} C_2^\alpha \varepsilon_{\alpha\beta} F_3^\beta\,,\\
\zeta_\alpha=\varepsilon_{\alpha\beta} C_2^\beta+h_\alpha\,;\qquad dh_\alpha-\varepsilon_{\alpha\beta} A^{\beta\gamma}h_\gamma=0 \label{halphaequation}\,.
\end{align}
These are the most general (local) solutions to those equations (recall that $(d+A)^2=0$ because $A$ is a flat connection). We proceed like in the M5 case and again use the identity $\delta C=\iota_{\delta x}d C+ d \iota_{\delta x} C$ (valid for any $C$ on $\Sigma_5$ that is the pullback of a form $C$ on the target $M$) to massage the boundary integrand into
\begin{align}
-d\left( \iota_{\delta x} C_4+\frac{1}{2} \iota_{\delta x} C_2^\alpha \varepsilon_{\alpha\beta} C_2^\beta+ \iota_{\delta x}C_2^\alpha h_\alpha\right)\nonumber\\
+\delta\left(C_4+\frac{1}{2} C_2^\alpha h_\alpha\right)+\frac{1}{2}h_\alpha \varepsilon^{\alpha\beta} \delta h_\beta +\frac{1}{2}\iota_{\delta x} A^{\alpha\beta}h_\alpha h_\beta\,.\label{iibboundaryterm2}
\end{align}

The analysis now bifurcates from the M5 case. We now assume that $A^{\alpha\beta}$ is associated to an $SL(2)/U(1)$ coset which we parameterise as in \eqref{iibtwistparametrisation} (without loss of generality). Then a short calculation shows that the general solution of \eqref{halphaequation} for $h_\alpha$ is
\be
h_2=e^{-\phi/2}h\,,\quad h_1=e^{\phi/2}(C_0 h + h')\,; \qquad dh=dh'=0\,.
\ee
In terms of the new variables $h,h'$ there are a number of cancellations in the last two terms in \eqref{iibboundaryterm2}:
\be
\frac{1}{2}h_\alpha \varepsilon^{\alpha\beta} \delta h_\beta +\frac{1}{2}\iota_{\delta x} A^{\alpha\beta}h_\alpha h_\beta=\frac{1}{2}(h'\delta h-h \delta h')\,.
\ee
Before we discuss boundary conditions on the pair $(h,h')$ we reexpress the potentials $C_2^\alpha, C_4$ in terms of potentials $C_2',B$ for the RR, and NSNS 3-forms $G_3,H$ and $C'_4$ for the RR 5-form $G_5$ (cf. \eqref{iibtwistparametrisation}):
\be
C^1_2=-e^{-\phi/2} B\,,\quad C_2^2=e^{\phi/2}(C'_2 + C_0 B)\,, \quad C_4=-C'_4 + B C_2'
\ee
where
\be
dB=H, \quad dC_2'=G_3+C_0H\,,\quad G_5=-F_5=dC'_4-H_3 C_2'\,.
\ee
for which
\be
C_2^\alpha h_\alpha= C_2' h - Bh'\,.
\ee
In terms of these quantities and after dropping the total derivatives, the boundary integrand is
\be
\delta\left(-C'_4+ B C'_2 +\frac 1 2(C_2' h - Bh')\right)+\frac{1}{2}(h'\delta h-h \delta h')\,.
\ee
Clearly one can regroup this term to again exhibit $SL(2)$ invariance. Interestingly this expression is rather different than previous proposals for $SL(2)$-covariant WZ terms for IIB D-branes \cite{Bergshoeff:2006gs,Cederwall:1997ab}.

Imposing a consistent boundary condition clearly involves breaking $SL(2)$. With the choice $h'=0$,
\begin{align}
\delta\left(S_{\rm IIB} + \int_{W_4}C'_4 - C_2'(B+h/2)\right)=0 &&\text{on-shell}\nonumber\\
\chi_\mu|_{W_4}=0\,, \zeta_1\,|_{W_4}=e^{\phi/2}(C_2'+C_0(B+h))\,, \zeta_2\,|_{W_4}= e^{-\phi/2}(B+h)\,,\qquad dh=0\,.
\end{align}

Comparison with the D3-brane WZ term \cite{Douglas:1995bn} (see e.g. \cite{deBoer:2010ud} appendix D) shows that on-shell we have $S_{\rm IIB}=-S_{\text{D3, WZ}}$ if we identify $h=2dV$ (for $V$ the usual D-brane gauge field) \emph{except} for the axion coupling $C_0(B+h/2)^2$, which is missing. This is mathematically consistent as $C_0(B+h/2)^2$ is gauge-invariant by itself (under gauge transformations of $B$, which are compensated by $h$). Since the only input that went into the construction of the topological field theory is the form of the generalised Lie derivative in the IIB construction of the exceptional tangent bundle, and given that the physical content of the generalised Lie derivative is the gauge transformations of type IIB supergravity, it is perhaps unsurprising that we fail to obtain the RR 0-form coupling because its presence cannot be deduced by considerations involving gauge-invariance alone. What is more disturbing is that there is clearly a coupling to the RR 0-form in $S_{\rm IIB}$ \eqref{iibbosonicaksz} through the $A^{\alpha\beta}\zeta_\alpha \zeta_\beta$ coupling; however, this cancels on-shell (use the $\zeta_\alpha$ equation of motion \eqref{zetaalphafieldequation} in $S_{\rm IIB}$ directly).


\section{Discussion}
Let us take stock of what we did. By rewriting the generalised Lie derivative of exceptional generalised geometry in terms of a graded Poisson bracket, we obtained, with essentially no work, an understanding of the intricate $L_\infty$-algebra structure governing generalised diffeomorphisms, and also the canonically associated topological field theories which turn out to be closely related to the Wess-Zumino terms defining the couplings of known \emph{non}-topological branes --- specifically the M5 and D3 --- to supergravity background fields. We emphasise again that the dg-manifold (i.e. $L_
\infty$-algebroid) structure is trivial to guess, at least for the M-theory, $E_6$ and IIB, $E_5$ cases: given that the ``Leibniz algebroid'' identity \eqref{loday} holds, for a symplectic form of degree $p$ it must be the case that generalised vectors lie in degree $p-1$. It is then a matter of arithmetic to arrange for the inclusion of the correct exceptional generalised tangent bundles in the cases we considered. Given that the corresponding dg-manifold structure for the Courant algebroid (i.e. string theory/T-duality) has been known for decades \cite{Roytenberg:2002nu}, the constructions in this paper settle a rather obvious open problem, identified as such in \cite{Deser:2018oyg}.

We should mention that we have only seriously considered local features of the construction. It remains to be checked that the dg-symplectic manifold correctly reproduces the features of the twisted exceptional generalised tangent bundle as in \cite{Hull:2007zu,Pacheco:2008ps} but that seems likely because twists act as symplectomorphisms (see section \ref{subection:twists}). Along the way one expects a nice description of the gerbe structure underlying the eleven-dimensional supergravity 3-form field to fall out. In \cite{Severa:2017oew} letter 7 \v Severa suggests (some) gerbes should be thought as principal $\mathbb R[n]$-bundles. In fact such bundles enter crucially in the constructions of this paper  and  the ``strange variables'' $\zeta,\zeta_\alpha$ are coordinates along their fibres.

Consider now extending to $E_7$ (for M-theory) or $E_6$ (for IIB). The untwisted $E_7$ generalised tangent bundle is
\be
TM\oplus \Lambda^2 T^\star M \oplus \Lambda^5 T^\star M \oplus (\Lambda^7 T^\star M\otimes T^\star M)\,.
\ee
The difference from the $E_6$ case is in the last summand. The mixed-symmetry term $\Lambda^7 T^\star M\otimes T^\star M$ is tricky: for any $\mathcal M$ of the form $T^\star[p]T[1]M\times R[n]$ the highest-form object in the space of functions at any fixed degree $r$ will be an $r$-form $C_{\mu_1\dots\mu_r}(x)\psi^{\mu_1}\dots\psi^{\mu_r}$. Introducing a second coordinate $\xi^\mu$ at degree 1 is problematic: one gets the mixed symmetry potentials in $\Lambda^{q}T^\star M\otimes \Lambda^{r-q}T^\star M$ for all $0\leq q\leq r$, which is much bigger than the exceptional tangent bundle. Similar issues occur for IIB because of the $SL(2)$-doublet of 5-forms in the $E_6$ tangent bundle. 

One can view this apparent obstruction in two ways. The optimistic point of view is that given the relation of the constructions we were able to write down to M5 and D3 branes, it must be the case that the topological field theories canonically-associated to the larger duality groups must be hitherto-unknown, exotic topological field theories of significant physical interest. The pessimistic point of view which is in particular relevant for the M-theory construction is that given that the $E_7$ construction will likely have a symplectic form of higher degree than 6, the corresponding topological field theory will be defined in 8 dimensions or higher, and there is no known physically-relevant brane, or at least no known brane that plays the same role as the M2- and M5-branes in M-theory. A possible candidate is the M9-brane \cite{Horava:1996ma,Horava:1995qa} i.e. the boundary of eleven-dimensional spacetime. A small hint that this makes sense is the fact that M5-branes can end on M9-branes \cite{Bergshoeff:2006bs}. This and other possibilities for branes are considered in \cite{Hull:1997kt,Bergshoeff:1998bs}. For type IIB the higher-dimensional D-branes are candidates but these tend to come in nontrivial $SL(2)$-multiplets \cite{Bergshoeff:2006gs} and it is not easy to see how the formalism in this paper can produce anything that is not an $SL(2)$-singlet.

One might wonder if the M2 brane fits into a dg-symplectic picture, given the distinguished role the M2, D3 and M5 branes play in string/M-theory. It (or rather its WZ coupling) should arise from a topological 3-brane. The M2 brane WZ coupling only involves an integral of the eleven-dimensional supergravity 3-form $C$ over the worldvolume, so clearly this case corresponds to the dg-symplectic manifold $\mathcal M=T^\star[3]T[1]M$ of degree 3. This is analysed in \cite{grutzmann2011h,Ikeda:2010vz} and more recently in \cite{Kokenyesi:2018ynq} where the connection to exceptional generalised geometry is pointed out. A function at degree $p-1=2$ corresponds to a section of $TM\oplus \Lambda^2 T^\star M$ i.e. a section of the untwisted $E_d$ exceptional generalised tangent bundle for $2\leq d \leq 4$. We did not analyse this case in detail on account of how similar it is to the string case. Along with the known physically-relevant $p=1,2$ cases we have therefore found interesting physics in dg-symplectic manifolds of all degrees up to 6 \emph{excluding} degree 5.

On the relation to brane physics: the M-theory construction is related to the M5-brane as unambiguously as one might possibly hope for, given that the field theory one obtains from the corresponding dg-symplectic manifold is in seven rather than six dimensions; the chiral gauge field on the M5-brane worldvolume emerges quite naturally from the seven-dimensional Chern-Simons term, and the seven-dimensional lagrangian we obtained directly is identical to one proposed by Kalkinnen and Stelle \cite{Kalkkinen:2002tk} based on careful examination of Dirac quantisation conditions for M-theory fluxes (and in turn closely related to the ``Hopf-WZ'' term of Intriligator \cite{Intriligator:2000eq}). One might wonder whether the full M5-brane lagrangian can be obtained from the topological field theory, like how the usual string sigma model can be obtained from the Courant sigma model on two-dimensional boundaries \cite{Severa:2016prq}. On the face of it one can fix a boundary condition for $\chi_\mu$ so that the 6-form $\chi_\mu dx^\mu$ equals the missing M5-brane Dirac-Born-Infeld term at least for certain metrics. This is somewhat strange though on account of the fact that the DBI term here is not necessarily to the worldvolume metric for which the chiral worldvolume field is (anti)self-dual, although it is conceivable that the master equation in the presence of a worldvolume boundary places further restrictions than the ones we found. A related issue is that the chiral field we obtained is always linearly self-dual whereas the chiral field on the M5-brane worldvolume is non-linearly self-dual for generic target backgrounds, but that does not seem like a fatal inconsistency because there exists a field redefinition in terms of a linearly self-dual field \cite{Sezgin:1998tm}.

For the relation to the D3-brane the situation is similar as far as obtaining the D3-brane DBI term is concerned. One difference compared to the M5-brane case is that the D3-brane WZ term we obtained omits the axion coupling $C_0 (B+dV)^2$. While this is certainly mathematically consistent (this coupling is gauge invariant independently of the rest of the Wess-Zumino term), it is rather surprising this is ``missed'' by the current approach, given that the presence of the axion coupling can be deduced by T-duality from type IIA \cite{Green:1996bh}: since the point of the exceptional tangent bundle is that there is a linear action of a U-duality group (as detailed in Table \ref{udualitytable}) --- and so in particular of the T-duality subgroup --- it is rather bizarre the theory does not ``know'' of that term.

A possible resolution of that conundrum is that the constructions we have written down, while motivated by U-duality-covariance, do not seem to possess it manifestly. The same is true of EGG by virtue of the fact that one works with $M$ as opposed to the extended spacetime $\tilde M$ of EFT, however in the dg-symplectic manifold context one cannot even do such things as e.g. write generalised tensors in arbitrary representations of the duality group as functions on $\mathcal M$ (but do recall we get the representations $R_n$ in the tensor hierarchy up to $\bar R_1$). This is also true to some extent in the dg-symplectic manifold construction for ordinary ($O(d,d)$) generalised geometry. Writing $e_M=(\psi^\mu,\chi_\mu)$ for the coordinates at degree 1 in that construction, one can only write down tensors in $\Lambda^r (TM\oplus T^\star M)$ in the space of functions at degree $r$. Similarly we can write tensors in $\Lambda^r E$ as functions in degree $5r$ in the M-theory construction and $3r$ in the type IIB construction. A difference to the $O(d,d)$ generalised geometry construction is that $TM\oplus T^\star M$ is self-dual so one effectively also has $\Lambda^r E^\star$, or, what amounts to the same thing, the $O(d,d)$-structure is encoded in the graded Poisson bracket at degree 1. For the constructions in this paper the $E_d$-invariant tensors are harder to see except for the $E_5={\rm Spin}(5,5)$ quadratic invariant, which is simply the graded Poisson bracket in degree $4$.

A point related to the lack of manifest $E_d$-covariance is the precise relation to the $L_\infty$-algebra of EFT as formulated in  the very recent work \cite{Cederwall:2018aab} by Cederwall and Palmkvist. The obvious differences are firstly that the $L_\infty$-algebra in that reference is much bigger (both in terms of length and of the size of the $\mathbb L_\ell$ modules), and secondly that all modules therein are $E_d$-multiplets. These differences arise from the ``ancillary ghosts'' necessary for $E_d$-covariance: the non-ancillary line $q=0$ in \cite{Cederwall:2018aab} Table 1 matches up precisely with our result, up to $\bar R_1$ (for our M-theory construction) and $\bar R_2$ (for our type IIB construction), beyond which our modules no longer form $E_d$ multiplets. The expectation is that the $L_\infty$-algebra of \cite{Cederwall:2018aab} matches up precisely with the one in this paper once a choice of M-theory or type IIB section is made. It is furthermore interesting to speculate to what extent our algebra of functions $C^\infty(\mathcal M)$ is related to the Borcherds superalgebra of \cite{Cederwall:2018aab} from which their $L_\infty$-algebra was derived. One commonality is that both the Poisson bracket here and the Lie algebra bracket in the Borcherds superalgebra coincide with the tensor hierarchy bullet product $\bullet$ for those levels in the $L_\infty$-algebras which agree with each other (modulo ``ancillary''). However the Poisson bracket generically fails to be ultralocal whenever functions at $C_{n\geq p}^\infty(\mathcal M)$ are involved (because these can involve $p_\mu$).

\bigskip

So far we mostly discussed generalisations to larger duality groups. An entirely orthogonal class of generalisations is motivated thusly: \emph{given that the constructions in this paper are directly analogous to \emph{exact} Courant algebroids (in their formulation as degree 2 dg-symplectic manifolds of the form $\mathcal M=T^\star[2]T[1] M$), how do non-exact Courant algebroids generalise?} An arbitrary Courant algebroid is simply a degree 2 dg-symplectic manifold \cite{Roytenberg:2002nu}, so one could simply consider arbitrary dg-symplectic manifolds of degrees 6 and 4 (generalising the M-theory and type IIB constructions respectively). We suspect that these generalisations are too wide and that extra conditions might be needed. For that reason consider the following very conservative generalisation of the M-theory construction
\be
\mathcal M=T^\star[6]T[1]M\times\mathbb{R}^n[3]\,,\quad \omega=dp_\mu dx^\mu- d\chi_\mu d\psi^\mu-\delta^{ij}d\zeta_i d\zeta_j\,,\qquad i=1,2\dots n\,.
\ee
The $\zeta^i$ bilinears define a ${\rm Spin}(n)$ subalgebra under the Poisson bracket. The hamiltonian $\Theta$ can now include a term $A^{ij}_\mu(x) \zeta_i\zeta_j \psi^\mu$ defining a ${\rm Spin}(n)$ connection on $M$ (flat when $(\Theta,\Theta)=0$). From the calculations of section \ref{section:branewzterms} it would appear we obtain a topological 6-brane theory that produces a multiplet of chiral forms $h_i$ (arising from $\zeta_i$) in the $\bf n$ of ${\rm Spin}(n)$ on 5-brane boundaries. It is probably too much to hope that this theory is directly related to the elusive nonabelian $(2,0)$ theory in six dimensions for gauge group ${\rm Spin}(n)$, but it could plausibly be an example of a theory for a nonabelian tensor multiplet, and even those are in short supply.

\bigskip

{\bf Acknowledgements:} {\small I am grateful for the interactions I had with Chris Hull, David Tennyson, and Daniel Waldram while writing this paper and also for helpful correspondence from Alec Barns-Graham, Chris Blair and especially Martin Cederwall and Jakob Palmkvist (with regards to their work \cite{Cederwall:2018aab}). I also had useful input from Urs Schreiber, Richard Szabo and Christian S\"amann. Finally (and belatedly) I would like to thank Charles Strickland-Constable for pointing out the error in the original \eqref{hamiltonianfromvectorfield}.

I am supported
by the EPSRC programme grant “New Geometric Structures from String Theory” (EP/K034456/1).}

\appendix
\section{Geometry on symplectic, graded supermanifolds}
\label{appendix:gradedsupermanifolds}

Let $\mathcal M$ be a graded supermanifold (see e.g. \cite{jozwikowski2016note} section 5): a supermanifold with globally-defined degree-counting vector field $\epsilon$ (also called the ``Euler'' vector field due to Euler's homogeneous function theorem) and transition functions preserving the degree. One can therefore always find homogeneous coordinates $z^a$: coordinates of definite degree $\deg z^a$,  which we will also denote as $a$. {\bf We assume the degree is correlated with supermanifold parity,} i.e.
\be
z^az^b=(-1)^{ab}z^b z^a\,.
\ee
This assumption reflects the fact we only consider bosonic theories in the physics sense (none of the fields at ghost number zero are parity-odd).

The space of functions on a graded supermanifold $\mathcal M$ splits into a sum of subspaces $C_n^\infty(\mathcal M)$ spanning all functions of degree $n$.

The graded supermanifolds we will explicitly consider are built using the degree-shifting construction from (graded) vector bundles. For instance, if $V$ is an ordinary vector bundle over an ordinary manifold $M$, $V[n]$ is the graded supermanifold given by shifting the degree of the fibre by $n$. Local homogeneous coordinates for $V[n]$ are therefore written $(x^\mu,v^I)$ where $x^\mu$ is a coordinate on $M$ (of degree zero) and $v^I$ is a coordinate on the fibre of degree $n$. In particular this implies $v^Iv^J=(-1)^n v^J v^I$. The degree shift can be used on graded vector bundles to obtain other ones, notably $T^\star[p]T[1]M$ with local coordinates $(x^\mu,\psi^\mu,\chi_\mu,p_\mu)$ of degrees $0,1,p-1$ and $p$ respectively. (This would be written $T^\star(TM[1])[p]$ in a more consistent notation.)

Unadorned derivatives $\partial_a$ are left derivatives, while $\partial^R_a$ are right derivatives:
\be
df=dz^a \partial_a f=\partial_a^R f dz^a
\ee
and vector fields $X=X^a \partial_a$ of definite degree $X$ act on functions $f,g\in C^\infty(\mathcal M)$ as ($\deg X:=\deg X^a-a,(-1)^f:=(-1)^{\deg f}$)
\begin{align}
X\cdot f&=X^a \partial_a f \\
X\cdot(fg)&=(X\cdot f) g + (-1)^{Xf} f (X\cdot g)\,.
\end{align}
One calculates that the graded commutator of vector fields is another vector field: $[X,Y]$ acts on functions as a derivation of degree $X+Y$ if we set
\be
\label{gradedcommutator}
[X,Y]\cdot f:= X\cdot(Y\cdot f)- (-1)^{XY} Y\cdot(X\cdot f)
\ee
and so $[Y,X]=(-1)^{1+XY}[X,Y]$. The space of vector fields is embedded inside a graded associative algebra under composition ($(XY)\cdot f:=X\cdot(Y\cdot f)$) and a short calculation gives the \emph{graded} Jacobi identity for any three vector fields $X_\alpha,X_\beta,X_\gamma$ of degrees $\alpha,\beta,\gamma$
\be
\label{superjacobi}
(-1)^{\alpha\gamma}[X_\alpha,[X_\beta,X_\gamma]] + (-1)^{\beta\alpha} [X_\beta,[X_\gamma,X_\alpha]] + (-1)^{\gamma\beta}[X_\gamma,[X_\alpha,X_\beta]]=0\,.
\ee
Equivalently,
\be
\label{superadjointrep}
{\rm ad}_{X_\alpha}[X_\beta,X_\gamma]= [{\rm ad}_{X_\alpha} X_\beta,X_\gamma]+(-1)^{\alpha\beta}[X_\beta, {\rm ad}_{X_\alpha}X_\gamma]\qquad ({\rm ad}_{X_\alpha} Y:=[X_\alpha,Y])\,.
\ee

We define differential forms as functions on $T[1]\mathcal M$ (the tangent bundle where the fibre coordinate $dz^a$ is declared to be of degree $\deg a+1$). The exterior derivative is the vector field $d$ of degree 1 on $T[1]\mathcal M$ acting on functions on $\mathcal M$ as above, and on the fibre as $d (dz^a)=0$.

We define the interior product or contraction $\iota_X$ as the degree $X-1$ vector field on $T[1]\mathcal M$ satisfying
\be
\iota_X dz^a =  X^a\,, \quad \iota_X z^a=0\implies \iota_X df= X \cdot f=X^a\partial_a f\,,
\ee
(this disagrees with e.g. \cite{Ikeda:2012pv,Roytenberg:2006qz}) and the Lie derivative by the magic formula
\be
\label{magic}
\mathcal L_X=[\iota_X,d]=\iota_X d +(-1)^X  d\iota_X
\ee
(where the commutator is the graded commutator). We calculate $\mathcal L_X f=X\cdot f$ for any function $f \in C^\infty(\mathcal M)$, so $\mathcal L_X$ is the correct extension of the Lie derivative to any form on $\mathcal M$: it is a degree $X+1-1=X$ vector field on $T[1]\mathcal M$. On vector fields we define $\mathcal L_X Y:=[X,Y]$.

The identity \eqref{superadjointrep} can then be used to prove the graded generalisations of all the usual identities since $\mathcal L_X, d, \iota_X$ are all vector fields on $T[1]\mathcal M$. For example, $[d,\mathcal L_X]=[d,[\iota_X,d]]=-[d,[\iota_X,d]]=0$. Slightly less trivially,
\be
[\mathcal L_X,\mathcal L_Y]=\mathcal L_{[X,Y]}
\ee
on all forms is proven by noting a) it is true by definition \eqref{gradedcommutator} on functions and b) it is true on all 1-forms $df$ since $[\mathcal L_X,\mathcal L_Y]$ moves past $d$. We will use the more obscure
\be
\label{miscmagic}
[\mathcal L_X,\iota_Y]=\iota_{[X,Y]}
\ee
later (trivially proved on 1-forms $df$).

In any system of homogeneous coordinates $z^a$ we have the following expression for the Euler vector field $\epsilon$:
\be
\epsilon=(\deg z^a)z^a\partial_a\,.
\ee
The degree of a function, vector field or differential form on $\mathcal M$ is simply the eigenvalue of the Lie derivative $\mathcal L_\epsilon$. $\epsilon$ itself has zero degree.
(Differential forms on $\mathcal M$ also have another degree associated to their interpretation as functions on $T[1]\mathcal M$. The difference between the two is form degree. So ``an $n$-form of degree $p$'' means $\epsilon$-degree $p$, total degree $n+p$.)

\subsection{Graded symplectic form and Poisson bracket}
Graded Poisson brackets $(\--,\--)$ of degree $-p$ are defined to satisfy the graded antisymmetry property
\be
\label{poissonantisymmetry}
(f,g)=(-1)^{1+(f+p)(g+p)} (g,f)
\ee
and the graded Jacobi identity
\be
\label{poissonjacobi}
(f,(g,h))= ((f,g),h)+ (-1)^{(f+p)(g+p)}(g,(f,h))
\ee
or equivalently
\be
(-1)^{(f+p)(h+p)}(f,(g,h)) + (-1)^{(g+p)(h+p)} (h,(f,g))+(-1)^{(f+p)(g+p)}(g,(h,f))=0\,.
\ee
These look strange due to factors of $p$ but are in fact equivalent to the usual super-Jacobi identity \eqref{superjacobi} for the corresponding hamiltonian vector fields; also, they define an ordinary graded Lie bracket on a ``degree-reflected'' space (see next subsection). For $p=-1$ the bracket is of degree $1$ and these identities characterise the BV antibracket.

A Poisson bracket of degree $-p$ is defined by a symplectic form of degree $p$: a closed 2-form of degree $p$ which is nondegenerate. We define hamiltonian vector fields $X_f$ and Poisson bracket $(f,g)$ for any two functions $f,g$ on $\mathcal M$ by
\be
\label{poissondef}
\iota_{X_f}\omega=(-1)^f df\,,\qquad (f,g)=\mathcal L_{X_f} g = X_f \cdot g\,.
\ee
Graded antisymmetry follows from (NB $\deg X_f=f-p$)
\be
(f,g)=\iota_{X_f}dg=(-1)^g \iota_{X_f}\iota_{X_g}\omega=(-1)^{1+(f+p)(g+p)} ((-1)^f\iota_{X_g}\iota_{X_f}\omega)
\ee
for which the extra sign in \eqref{poissondef} is crucial. The graded Jacobi identity in the form \eqref{poissonjacobi} is trivially verified by a) commuting the Lie derivatives in $(f,(g,h))=\mathcal L_{X_f} \mathcal L_{X_g} h$ using \eqref{superadjointrep} and b) using the identity
\be
\label{pbhomomorphism}
[X_f,X_g]=X_{(f,g)}\,.
\ee
This is proven using \eqref{miscmagic} as follows
\begin{align}
\iota_{[X_f,X_g]}\omega&=(\mathcal L_{X_f} \iota_{X_g} - (-1)^{X_f(X_g+1)} \iota_{X_g} \mathcal L_{X_f})\omega\\
&=(-1)^{f+p}d \iota_{X_f}\iota_{X_g}\omega\\
&=(-1)^{f+p+g}d\iota_{X_f} dg\\
&=(-1)^{f+g-p}d(f,g)=\iota_{X_{(f,g)}}\omega
\end{align}
recalling that the degree of $(f,g)$ is $f+g-p$.

We will use the fact symplectic forms of degree $p\neq 0$ are always exact. In  fact there is a canonical symplectic potential $\vartheta$ satisfying $-d\vartheta=\omega$, obtained using a trick involving the Euler vector field: since $\mathcal L_\epsilon \omega=p\omega=d(\iota_\epsilon\omega)$,
\be
\label{canonicalsymplecticpotential}
\vartheta=-\frac{1}{p}\iota_\epsilon\omega\,.
\ee
There is the following formula for the hamiltonian $f_X$ associated to any vector field $X$ (of degree $X\neq -p$) leaving the symplectic form $\omega$ of degree $p$ invariant:
\be
\label{hamiltonianfromvectorfield}
f_X=\frac{(-1)^p p}{X+p} \iota_X\vartheta\,,
\ee
proven using the Euler vector field (where $\vartheta$ is the canonical symplectic potential defined above).

We can now finally state
\begin{definition}
\label{dgsymplecticmanifold}
A \emph{dg-symplectic manifold} of degree $p\neq 0$ is a graded symplectic supermanifold $\mathcal M$ with symplectic form $\omega$ of degree $p$ and distinguished hamiltonian $\Theta\in C^\infty(\mathcal M)$ of degree $p+1$. Equivalently, it is a graded symplectic supermanifold with homological vector field $X$ ($\deg X=1, \mathcal L_X X=0$)  that preserves the symplectic form i.e. $\mathcal L_X\omega=0$.

We assume that there are no coordinates of negative degree on $\mathcal M$. With this assumption, dg-symplectic manifolds of degree $p$ are also known as: \emph{symplectic Lie $p$-algebroids} \cite{Ritter:2015ffa}, \emph{NPQ-manifolds} \cite{Alexandrov:1995kv}, \emph{$\Sigma_p$-manifolds} \cite{severa2001some}, \emph{$Qp$-manifolds} \cite{Gruetzmann:2014ica} and finally \emph{$L_\infty$-algebroids} \cite{Sati:2009ic}.
\end{definition}

\subsubsection{Infinitesimal automorphisms}
\label{automorphismsgradedPoisson}
We define an infinitesimal automorphism of a graded Poisson bracket as any vector field $X$ on $\mathcal M$ of degree zero deriving the Poisson bracket:
\begin{align}
X\cdot(f,g)=(X\cdot f,g)+(f,X \cdot g) && \forall f,g\in C^\infty(\mathcal M)\,.
\end{align}
This is the condition obtained by differentiating a graded Poisson algebra automorphism at the identity; accordingly, $\exp(X)$ defines a (finite) automorphism if e.g. $X$ is nilpotent.

For $p\neq 0,1$ it is trivial to prove that \emph{all infinitesimal automorphisms $X$ are inner: $X=(f_X,\--)$ for $f_X$ as given by \eqref{hamiltonianfromvectorfield}}. It follows from showing $\mathcal L_X \omega=0$. (Sketch: Write $\mathcal L_X(f,g)=[\mathcal L_X,\mathcal L_{X_f}]g+ \mathcal L_{X_f}\mathcal L_X g=\iota_{[X,X_f]}g+(f,X\cdot g)$, then massage the first term with \eqref{miscmagic}).

\subsubsection{Explicit coordinate expressions}
We write the symplectic form $\omega $ as
\be
\omega=\frac{1}{2}dz^a\omega_{ab} dz^b
\ee
from which follows the symmetry property
\be
\label{omegasymmetry}
\omega_{ba}=(-1)^{1+ab + p(a+b)}\omega_{ab}\,.
\ee
If we define the Poisson bivector $\omega^{ca}$ as $\omega^{ca}\omega_{ab}=\delta^c_b$ and calculate $\iota_X \omega= X^a \omega_{ab} dz^b$ for any vector field $X$, \eqref{poissondef} gives
\be
\boxed{X_f^a=(-1)^f \partial_c^R f \omega^{ca}\,,\quad (f,g)= (-1)^f \partial_a^R f \omega^{ab} \partial_b g\,.}
\ee
Therefore
\be
(z^a,z^b)=(-1)^{a}\omega^{ab}\,.
\ee
This last formula might appear inconsistent with the symmetry of the Poisson bracket; however that is because $\omega^{ab}$ \emph{does not} have the same symmetry as $\omega_{ab}$ but rather
\be
\omega^{ba}=(-1)^{1+ab+p(a+b) + (a+b+p)}\omega^{ab}\,.
\ee
We note however that the extra sign $(-1)^{a+b+p}$ vanishes in Darboux coordinates so this subtlety never comes up.

For the canonical symplectic potential $d\vartheta=-\omega$ obtained from the Euler vector field we get
\be
\vartheta=-\frac{1}{p} (\deg z^a)z^a \omega_{ab} dz^b\,.
\ee

All constructions considered in this paper have even symplectic forms and so the form of $\omega$  is near-universal. We use the Poisson brackets
\begin{align}
(x^\mu,p_\nu)=\delta^\mu_\nu=-(p_\nu,x^\mu)\,,\quad (\psi^\mu,\chi_\mu)=\delta^\mu_\nu=(\chi_\nu,\psi^\mu)\,,\\
(\zeta,\zeta)=1\,,\quad (\zeta_\alpha,\zeta_\beta)=\varepsilon_{\alpha\beta}\,.
\end{align}
whence using $(z^a,z^b)\omega_{bc}=(-1)^a \delta^a_c$ (where some or all of the $\zeta,\zeta_\alpha$ terms are absent depending on the case)
\be
\omega=dp_\mu dx^\mu+\frac{1}{2} d\zeta_\alpha \varepsilon^{\alpha\beta}d\zeta_\beta- d\chi_\mu d\psi^\mu-\frac{1}{2}d\zeta d\zeta
\ee
and
\begin{itemize}
	\item $\mathcal M=T^\star[2]T[1]M$ (exact Courant algebroid):
	\be
	\vartheta=-p_\mu dx^\mu+\frac{1}{2}(\psi^\mu d\chi_\mu + \chi_\mu d\psi^\mu)
	\ee
	\item $\mathcal M=T^\star[4]T[1]M\times \mathbb{R}^2[2]$ (type IIB construction in this paper):
	\be
	\vartheta=-p_\mu dx^\mu +\frac{1}{4}(\psi^\mu d\chi_\mu + 3 \chi_\mu d\psi^\mu)+\frac{1}{2}\zeta_\alpha \varepsilon^{\alpha\beta}d\zeta^\beta\qquad(\varepsilon_{\alpha\beta}\varepsilon^{\gamma\beta}=\delta^\gamma_\alpha)
	\ee
	\item $\mathcal M=T^\star[6]T[1]M\times \mathbb{R}[3]$ (M-theory construction in this paper):
	\be
	\vartheta=-p_\mu dx^\mu+\frac{1}{6}(\psi^\mu d\chi_\mu + 5 \chi_\mu d\psi^\mu)+\frac{1}{2}\zeta d \zeta
	\ee
\end{itemize}

\subsection{The associated differential graded Lie algebra and $L_\infty$-algebra}
\label{appendix:linfinity}
\begin{definition} A \emph{differential graded Lie algebra (dgLa)} is a graded Lie algebra $L$ with a differential $Q$ of degree $-1$ such that $Q[a,b]=[Qa,b]+(-1)^{\deg a}[a,Qb]$. (The Lie bracket $[\--.\--]$ is degree zero and graded-anticommutative.)
\end{definition}

For any dg-symplectic manifold $\mathcal M$ of degree $p$ as above (but not necessarily non-negatively graded), the space $C^\infty(\mathcal M)$ has a natural dgLa structure \emph{after a reflection in degree:} the space of functions of degree $n$ on $\mathcal M$ $C_n^\infty(\mathcal M)$ is identified with the subspace $L_{p-n}$ of degree $(p-n)$ in the dgLa $L$. The Poisson bracket is then identified with the Lie bracket and \eqref{poissonantisymmetry} and \eqref{poissonjacobi} express the graded antisymmetry and super Jacobi identities respectively. The Lie bracket is degree zero in the new grading since $[\--,\--]:L_m\otimes L_n\to L_{m+n} \leftrightarrow (\--,\--):C^\infty_{p-m}(\mathcal M)\otimes C^\infty_{p-m}(\mathcal M)\to C^\infty_{p-m-n}(\mathcal M)$. The differential $Q$ is identified with $X_\Theta$ which is degree -1 on $L$.

As was pointed out in \cite{Ritter:2015ffa},
\begin{center}
For any dg-symplectic manifold $\mathcal M$ one gets a canonically associated $L_\infty$-algebra structure on the space of functions of degrees $n<p$.
\end{center}
\noindent This follows from the result in the note \cite{getzler2010higher} (which is equivalent to the earlier \cite{fiorenza2007structures}) which associates an $L_\infty$-algebra to any dgLa $L$, constructed from the positively-graded subspace $\oplus L_{m>0}$. Here the dgLa where this comes from is the one associated to the dg-symplectic manifold $\mathcal M$ in the previous paragraph. If $\mathcal M$ is non-negatively graded as is the case for the constructions in this paper we in fact obtain an $L_p$-algebra (an $L_\infty$-algebra where $L_{p+n}$ is zero for $n>0$). Unfortunately a straightforward proof of the proposition which does not otherwise rely on brutal calculations seems to be unavailable.

To exhibit said $L_\infty$-algebra we first give the definition in the convention of \cite{getzler2010higher}:
\begin{definition} An $L_\infty$-algebra is a graded vector space $\mathbb L$ with $n$-ary graded symmetric linear operators $\{\dots \}$ (so $\{a_1,a_2,\dots a_n \}\in\mathbb L\, \forall a_1,a_2\dots a_n \in \mathbb L$), all of degree $-1$, satisfying the following \emph{$k$-th Jacobi rule} for all $k\geq 0$:
\begin{align}
\label{kjacobirule}
\sum_{n=0}^k \sum_{\pi\in S_n} \binom{k+1}{n+1} \frac{(-1)^{\varepsilon(\pi)}}{(k+1)!} \big\{\{a_{\pi_1},\dots a_{\pi_{n+1}} \}, a_{\pi_{n+2}},a_{\pi_{n+3}},\dots a_{\pi_{k+1}} \big\}=0
\end{align}
where $(-1)^{\varepsilon(\pi)}$ is the usual graded symmetric sign: $a_{\pi_1}a_{\pi_2}\dots a_{\pi_{k+1}}=(-1)^{\varepsilon(\pi)} a_1a_2\dots a_{k+1}$.

\end{definition}
In the above definition, and in most of the literature, $L_\infty$-algebras have no 0-ary bracket.

\begin{prop} \cite{getzler2010higher}
The $L_\infty$-algebra $\mathbb L$ associated to the dgLa $L$ with differential $Q$ is given by the graded vector space $\mathbb L=\oplus L_{m>0}$ with brackets
\begin{align}
\{a\}&= Qa \qquad(\text{or zero for $\deg a=1$})\\
\{a_1,\dots a_{n+1}\}&=\frac{(-1)^{n}}{n!} B_{n}\sum (-1)^\varepsilon \Big[\dots \big[ [Q a_{\pi_1}-\{a_{\pi_1}\},a_{\pi_2}],a_{\pi_3}\big]\dots, a_{\pi_{n+1}}\Big]
\end{align}
where the sum is over all permutations $\pi$ and $(-1)^\varepsilon$ is as in the previous sum; $B_{n\geq 0}$ are the Bernoulli numbers $1,-1/2,1/6,0,\dots$, vanishing for odd $i\geq 3$. (NB that $Qa-\{a\}$ is nonzero, but only for $\deg a=1$).
\label{linftyfromdgla}
\end{prop}

For the $L_\infty$-algebra $\mathbb L$ associated to a dg-symplectic manifold of degree $p$ this simply says that the $n$-ary brackets always take the form $\Big(\dots\big((\Theta,A),f_1\big)\dots ,f_{n-1}\Big)$ where $A$ is a function of degree $p-1$. In the context of the exceptional generalised geometry constructions of this paper such expressions are interpreted as bullet products of the form $L_A f_1\bullet f_2\bullet \dots \bullet f_{n-1}$, where $L_A f_1$ is the generalised Lie derivative of $f_1$; all elements of such $\mathbb L$ live in the tensor hierarchy.

\section{AKSZ topological field theory}
\label{appendix:aksz}
The AKSZ construction \cite{Alexandrov:1995kv} associates to a graded symplectic supermanifold $\mathcal M$ of degree $p$ as in the previous section a field theory describing a \emph{closed} topological $p$-brane embedded in $\mathcal M$. One constructs a solution $S$ to the Batalin-Vilkovisky master equation
\be
(S,S)_{\rm BV}=0
\ee
where $(\--,\--)_{\rm BV}$ is the \emph{antibracket}, a certain graded Poisson bracket of degree $-1$. Both $S$ and the antibracket arise from structures on $\mathcal M$ and a certain supermanifold $\mathcal N$ whose body is the $p$-brane worldvolume $\Sigma$ (an ordinary manifold) as we review below.

We will only treat the usual case $\mathcal N=T[1]\Sigma$. Functions $\rho$ on $T[1]\Sigma$ are identified with polyforms on $N$ and can be integrated using the usual Berezin integral --- which we will abbreviate as $\int_{T[1]\Sigma}$ --- that picks out the top form component:
\begin{align}
\int_{T[1]\Sigma}\rho:&= \int_\Sigma d^{p+1}\sigma  \, \frac{1}{(p+1)!} \rho(\sigma)_{\alpha_1\dots \alpha_{p+1}} \theta^{\alpha_1}\theta^{\alpha_2}\cdots \theta^{\alpha_{p+1}}\, d\theta^{{p+1}}\cdots d \theta^{1}\\
&= \int_\Sigma d^{p+1}\sigma  \, \frac{1}{(p+1)!} \rho(\sigma)_{\alpha_1\dots \alpha_{p+1}} \varepsilon^{\alpha_1\dots \alpha_{p+1}}
\end{align}
where $\sigma^\alpha$ are coordinates on $\Sigma$ and $\theta^\alpha$ are coordinates on the fibre of $T[1]\Sigma$. The exterior derivative on $\Sigma$ is thus $d=\theta^\alpha\partial_\alpha$ (we identify $d\sigma^\alpha=\theta^\alpha$).

To find the fields of the AKSZ model, take each local homogeneous coordinate $z^a$ on $\mathcal M$ and promote it to a form $z^a(\sigma)$ on $\Sigma$ of rank equal to its degree. The ghost/antifield sector is obtained by promoting $z^a$ further to what should be thought of as a ``BV superfield'' $z^a(\sigma,\theta)$ of ghost number equal to the degree of $z^a$ as a coordinate on $\mathcal M$, while declaring $\theta^\alpha$ to have ghost number $+1$. Thus for any function $f \in C^\infty(\mathcal M)$ we have a corresponding function on $T[1]\Sigma$:
\be
f(\sigma,\theta)=f^0(\sigma) +f^1_\alpha(\sigma) \theta^\alpha +\dots + f^{\deg f}_{\alpha_1\dots\alpha_{\deg f}}(\sigma)\theta^{\alpha_1} \cdots \theta^{\alpha_{\deg f}}+\dots
\ee
where $f^0$ is a function on $\Sigma$ of ghost number $\deg f$, $f^1_\alpha$ corresponds to a 1-form on $\Sigma$ of ghost number $\deg f-1$, $f^{\deg f}_{\alpha_1\dots\alpha_{\deg f}}$ are the components of a $(\deg f)$-form on $\Sigma$ of ghost number zero, etc. \footnote{The promotion of $f$ to a BV superfield involving the coefficient functions $f^i$ for $i\neq \deg f$ (which carry intrinsic nonzero ghost number) might seem strange. The reason is explained in the proof of Proposition 2.8 of \cite{Roytenberg:2006qz} which defines the space ${\rm Maps}(T[1]\Sigma,\mathcal M)$ which is the correct definition of the space of fields. In short: if one wants ${\rm Maps}(\{\text{point}\},\mathcal M)$ to be the same as $\mathcal M$ one needs the ghosts.}

For any form on the target space $\mathcal M$ one obtains a form on the space of fields ${\rm Maps}(T[1] \Sigma,\mathcal M)$ (at least if certain analytic assumptions are invoked; we will ignore such issues as is customary in this context). This is called the ``transgression'' procedure in the mathematical literature on the AKSZ sigma model whereas in physics it has been used before (see e.g. \cite{Polchinski:1985zf}) with no name at all. For the symplectic form $\omega=dz^a\omega_{ab}dz^b/2$ on $\mathcal M$ we get
\be
\Omega:=\int_{T[1]\Sigma} \, \frac{1}{2}\delta z^a \omega_{ab} \delta z^b
\ee
where $\delta$ is the (left) exterior derivative on the space of fields, i.e. for any local functional $F[z]$ in terms of the left functional derivative
\be
\delta F= \int_{T[1]\Sigma} \, \delta z^a \frac{\delta F}{\delta z^a}\,.
\ee
Clearly $\Omega$ is of degree $p-(p+1)=-1$ (as a 2-form on ${\rm Maps}(T[1] \Sigma,\mathcal M)$, where the degree is by ghost number) and thus the corresponding Poisson bivector yields an antibracket. It gives the antibracket $(\--,\--)_{BV}$. Then one can use the formulas in the previous subsection to derive and prove properties of the antibracket (NB we defined the Berezin integral so the $d\theta$ act from the right; there are no funny signs). For example the antibracket of two local functionals $F,G$ is
\be
\label{antibracket}
(F,G)_{\rm BV}:=\mathcal L_{X_F} G=\int_{T[1]\Sigma} (-1)^F \frac{\delta^R F}{\delta z^a}\omega^{ab}\frac{\delta G}{\delta z^b}\,.
\ee

The BV master action for the AKSZ sigma model is simply a sum of hamiltonians: the hamiltonian generating the action of the $\Sigma$ exterior derivative $d$ and the hamiltonian $\Theta$ on target space $\mathcal M$, both acting on superfields (embeddings of $T[1]\Sigma$ in $\mathcal M$). The latter term is easy to write down and is simply $\int_{T[1]\Sigma}(-1)^{p+1} \Theta$. The former can be worked out from formula \eqref{hamiltonianfromvectorfield} keeping in mind that $d$ as a vector field on ${\rm Maps}(T[1]\Sigma,\mathcal M)$\footnote{Of course $\mathcal M$ also has an exterior derivative which we also denoted $d$. Hopefully it is clear which is which.} is defined by $\mathcal L_d F=\iota_d \delta F=\int_{T[1]\Sigma} dz^a \delta F/\delta z^a$ and we get 
\be
\int_{T[1]\Sigma} (-1)^{p} \vartheta_adz^a =\int_{T[1]\Sigma} (-1)^{ap} dz^a \vartheta_a \qquad (\vartheta=\vartheta_a dz^a \,, \omega=-d\vartheta \text{ on $\mathcal M$})
\ee
Check: this expression is of degree zero so the hamiltonian vector field $(\int(-1)^{ap}  dz^a \vartheta_a,\--)_{BV}$ is indeed of degree 1. To verify we first need the bracket
\be
(z^a,F)_{\rm BV}=(-1)^a \omega^{ab} \frac{\delta F}{\delta z^b}
\ee
for any local functional $F[z]$. Now let $F$ have ghost number $1\mod 2$ so $(\--,F)_{\rm BV}$ doesn't pick up minuses when moved past things\footnote{This is in fact without loss of generality for this calculation, as we can simply hit an $F$ of ghost number zero $\mod 2$ with a formal constant of ghost number $1$.}. Then
\begin{align}
\left(\int_{T[1]\Sigma} (-1)^{ap} dz^a \vartheta_a,F\right)_{\rm BV}&=\int_{T[1]\Sigma} (-1)^{ap} ( dz^a (\vartheta_a,F)_{\rm BV} + d(z^a,F)_{\rm BV} \vartheta_a)\\
&=\int_{T[1]\Sigma}(-1)^{1+ap+a}(z^a,F)_{\rm BV}dz^b((-1)^{1+ab+p(a+b)}\partial_a\vartheta_b + \partial_b\vartheta_a)\\
&=\int_{T[1]\Sigma} (-1)^{ap+a}(z^a,F)_{\rm BV} dz^b \omega_{ba}\\
&=\int_{T[1]\Sigma} dz^a \frac{\delta F}{\delta z^a}=\mathcal L_d F
\end{align}
To get to the second line we used integration by parts (recall $\Sigma$ is assumed closed). The third line follows from the second using $d\vartheta=d(\vartheta_a dz^a)=-\omega=-dz^a \omega_{ab}dz^b/2$ on $\mathcal M$ (NB the symmetry \eqref{omegasymmetry} of $\omega_{ab}$ under $a\leftrightarrow b$).

Let us write down the BV master action for the AKSZ sigma model. Since the hamiltonian for $d$ is the one containing any derivatives one might as well call it the ``kinetic'' part $S_{kin}$ of the action, and call the hamiltonian for $(\Theta,\--)$ the ``potential'' part $S_{pot}$. The AKSZ master action is
\be
\label{AKSZ}
S=S_{kin}-(-1)^{p+1}S_{pot}\,,\qquad S_{kin}= (-1)^{p}\int_{T[1]\Sigma}  \vartheta_a dz^a\,,\quad S_{pot}=(-1)^{p+1}\int_{T[1]\Sigma} \Theta\,.
\ee
(We explain the relative coefficient choice).
We now check the master equation. The previous calculation proves $(S_{kin},z^a)_{\rm BV}=dz^a$ (NB $d$ here is on $\Sigma$!). Therefore
\begin{align}
(S,S)_{\rm BV}&=(S_{kin},S_{kin})_{\rm BV}& &-2(-1)^{p+1}(S_{kin},S_{pot})_{\rm BV}& &+(S_{pot},S_{pot})_{\rm BV}\\
&= \int_{T[1]\Sigma}(-1)^{ap+a+1} dz^a d\vartheta_a & &-2  d \Theta& &+ (\Theta,\Theta)
\end{align}
where the first two terms vanish because $\partial \Sigma$ is empty and the last term involves $(\Theta,\Theta)$ on $\mathcal M$ which is zero by assumption. The value of the relative coefficient between $S_{kin}$ and $S_{pot}$ is irrelevant for the purpose of solving the master equation.

With our choice of relative coefficient the equation of motion is
\be
dz^a=(\Theta,z^a)\,.
\ee
where the bracket on the right-hand side is the Poisson bracket on $\mathcal M$. (Calculation sketch: $(-1)^{p+1} \delta S=\int\delta z^b(\omega_{ba}dz^a-(-1)^{p+1}\partial_b \Theta)$ and $\omega^{ab}\partial_b\Theta=(-1)^a (z^a,\Theta)=(-1)^{p+1}(\Theta,z^a)$)
Therefore the solutions are ``dg maps'' between $T[1] \Sigma$ and $\mathcal M$.

\bibliography{NewBib}

\providecommand{\href}[2]{#2}\begingroup\raggedright\begin{thebibliography}{10}

\bibitem{Hull:1994ys}
C.~Hull and P.~Townsend, {\it {Unity of superstring dualities}},  {\em
  Nucl.Phys.} {\bf B438} (1995) 109--137,
  [\href{http://arxiv.org/abs/hep-th/9410167}{{\tt hep-th/9410167}}].

\bibitem{Hull:2007zu}
C.~Hull, {\it {Generalised Geometry for M-Theory}},  {\em JHEP} {\bf 0707}
  (2007) 079, [\href{http://arxiv.org/abs/hep-th/0701203}{{\tt
  hep-th/0701203}}].

\bibitem{Pacheco:2008ps}
P.~P. Pacheco and D.~Waldram, {\it {M-theory, exceptional generalised geometry
  and superpotentials}},  {\em JHEP} {\bf 0809} (2008) 123,
  [\href{http://arxiv.org/abs/0804.1362}{{\tt arXiv:0804.1362}}].

\bibitem{Aldazabal:2013mya}
G.~Aldazabal, M.~Gra{\~n}a, D.~Marqu{\'e}s, and J.~Rosabal, {\it {Extended
  geometry and gauged maximal supergravity}},  {\em JHEP} {\bf 1306} (2013)
  046, [\href{http://arxiv.org/abs/1302.5419}{{\tt arXiv:1302.5419}}].

\bibitem{Berman:2010is}
D.~S. Berman and M.~J. Perry, {\it {Generalized Geometry and M theory}},  {\em
  JHEP} {\bf 1106} (2011) 074, [\href{http://arxiv.org/abs/1008.1763}{{\tt
  arXiv:1008.1763}}].

\bibitem{Berman:2011jh}
D.~S. Berman, H.~Godazgar, M.~J. Perry, and P.~West, {\it {Duality Invariant
  Actions and Generalised Geometry}},  {\em JHEP} {\bf 1202} (2012) 108,
  [\href{http://arxiv.org/abs/1111.0459}{{\tt arXiv:1111.0459}}].

\bibitem{Berman:2011cg}
D.~S. Berman, H.~Godazgar, M.~Godazgar, and M.~J. Perry, {\it {The Local
  symmetries of M-theory and their formulation in generalised geometry}},  {\em
  JHEP} {\bf 1201} (2012) 012, [\href{http://arxiv.org/abs/1110.3930}{{\tt
  arXiv:1110.3930}}].

\bibitem{Berman:2012vc}
D.~S. Berman, M.~Cederwall, A.~Kleinschmidt, and D.~C. Thompson, {\it {The
  gauge structure of generalised diffeomorphisms}},  {\em JHEP} {\bf 1301}
  (2013) 064, [\href{http://arxiv.org/abs/1208.5884}{{\tt arXiv:1208.5884}}].

\bibitem{Hohm:2013vpa}
O.~Hohm and H.~Samtleben, {\it {Exceptional Field Theory I: $E_{6(6)}$
  covariant Form of M-Theory and Type IIB}},  {\em Phys.Rev.} {\bf D89} (2014)
  066016, [\href{http://arxiv.org/abs/1312.0614}{{\tt arXiv:1312.0614}}].

\bibitem{Hohm:2013uia}
O.~Hohm and H.~Samtleben, {\it {Exceptional Field Theory II: E$_{7(7)}$}},
  {\em Phys.Rev.} {\bf D89} (2014) 066017,
  [\href{http://arxiv.org/abs/1312.4542}{{\tt arXiv:1312.4542}}].

\bibitem{Hohm:2014fxa}
O.~Hohm and H.~Samtleben, {\it {Exceptional Field Theory III: E$_{8(8)}$}},
  {\em Phys.Rev.} {\bf D90} (2014) 066002,
  [\href{http://arxiv.org/abs/1406.3348}{{\tt arXiv:1406.3348}}].

\bibitem{Hohm:2015xna}
O.~Hohm and Y.-N. Wang, {\it {Tensor Hierarchy and Generalized Cartan Calculus
  in SL(3)$\times$SL(2) Exceptional Field Theory}},
  \href{http://arxiv.org/abs/1501.01600}{{\tt arXiv:1501.01600}}.

\bibitem{Abzalov:2015ega}
A.~Abzalov, I.~Bakhmatov, and E.~T. Musaev, {\it {Exceptional field theory:
  $SO(5,5)$}},  {\em JHEP} {\bf 06} (2015) 088,
  [\href{http://arxiv.org/abs/1504.01523}{{\tt arXiv:1504.01523}}].

\bibitem{Musaev:2015ces}
E.~T. Musaev, {\it {Exceptional field theory: $SL(5)$}},  {\em JHEP} {\bf 02}
  (2016) 012, [\href{http://arxiv.org/abs/1512.02163}{{\tt arXiv:1512.02163}}].

\bibitem{Berman:2015rcc}
D.~S. Berman, C.~D.~A. Blair, E.~Malek, and F.~J. Rudolph, {\it {An Action for
  F-theory: $\mathrm{SL}(2) \times \mathbb{R}^+$ Exceptional Field Theory}},
  \href{http://arxiv.org/abs/1512.06115}{{\tt arXiv:1512.06115}}.

\bibitem{Baraglia:2011dg}
D.~Baraglia, {\it {Leibniz algebroids, twistings and exceptional generalized
  geometry}},  {\em J. Geom. Phys.} {\bf 62} (2012) 903--934,
  [\href{http://arxiv.org/abs/1101.0856}{{\tt arXiv:1101.0856}}].

\bibitem{lupercio2014t}
E.~Lupercio, C.~Rengifo, and B.~Uribe, {\it T-duality and exceptional
  generalized geometry through symmetries of dg-manifolds},  {\em Journal of
  Geometry and Physics} {\bf 83} (2014) 82--98.

\bibitem{Lada:1992wc}
T.~Lada and J.~Stasheff, {\it {Introduction to SH Lie algebras for
  physicists}},  {\em Int. J. Theor. Phys.} {\bf 32} (1993) 1087--1104,
  [\href{http://arxiv.org/abs/hep-th/9209099}{{\tt hep-th/9209099}}].

\bibitem{Cederwall:2018aab}
M.~Cederwall and J.~Palmkvist, {\it {$L_\infty$ algebras for extended geometry
  from Borcherds superalgebras}},  \href{http://arxiv.org/abs/1804.04377}{{\tt
  arXiv:1804.04377}}.

\bibitem{roytenberg1998courant}
D.~Roytenberg and A.~Weinstein, {\it Courant algebroids and strongly and
  strongly homotopy lie algebras},  {\em Letters in Mathematical Physics} {\bf
  46} (1998), no.~1 81--93.

\bibitem{Deser:2016qkw}
A.~Deser and C.~Saemann, {\it {Extended Riemannian Geometry I: Local Double
  Field Theory}},  \href{http://arxiv.org/abs/1611.02772}{{\tt
  arXiv:1611.02772}}.

\bibitem{Hohm:2017pnh}
O.~Hohm and B.~Zwiebach, {\it {$L_{\infty}$ Algebras and Field Theory}},  {\em
  Fortsch. Phys.} {\bf 65} (2017), no.~3-4 1700014,
  [\href{http://arxiv.org/abs/1701.08824}{{\tt arXiv:1701.08824}}].

\bibitem{Ritter:2015ffa}
P.~Ritter and C.~Saemann, {\it {Automorphisms of Strong Homotopy Lie Algebras
  of Local Observables}},  \href{http://arxiv.org/abs/1507.00972}{{\tt
  arXiv:1507.00972}}.

\bibitem{Sati:2009ic}
H.~Sati, U.~Schreiber, and J.~Stasheff, {\it {Differential twisted String and
  Fivebrane structures}},  {\em Commun. Math. Phys.} {\bf 315} (2012) 169--213,
  [\href{http://arxiv.org/abs/0910.4001}{{\tt arXiv:0910.4001}}].

\bibitem{kosmann1996poisson}
Y.~Kosmann-Schwarzbach, {\it From poisson algebras to gerstenhaber algebras},
  in {\em Annales de l'institut Fourier}, vol.~46, pp.~1243--1274, Chartres:
  L'Institut, 1950-, 1996.

\bibitem{voronov2005higher}
T.~Voronov, {\it Higher derived brackets and homotopy algebras},  {\em Journal
  of Pure and Applied Algebra} {\bf 202} (2005), no.~1-3 133--153.

\bibitem{Roytenberg:2002nu}
D.~Roytenberg, {\it {On the structure of graded symplectic supermanifolds and
  Courant algebroids}},  in {\em {Workshop on Quantization, Deformations, and
  New Homological and Categorical Methods in Mathematical Physics Manchester,
  England, July 7-13, 2001}}, 2002.
\newblock \href{http://arxiv.org/abs/math/0203110}{{\tt math/0203110}}.

\bibitem{Roytenberg:2006qz}
D.~Roytenberg, {\it {AKSZ-BV Formalism and Courant Algebroid-induced
  Topological Field Theories}},  {\em Lett. Math. Phys.} {\bf 79} (2007)
  143--159, [\href{http://arxiv.org/abs/hep-th/0608150}{{\tt hep-th/0608150}}].

\bibitem{courant1990dirac}
T.~J. Courant, {\it Dirac manifolds},  {\em Transactions of the American
  Mathematical Society} {\bf 319} (1990), no.~2 631--661.

\bibitem{liu1995manin}
Z.-J. Liu, A.~Weinstein, and P.~Xu, {\it Manin triples for lie bialgebroids},
  {\em arXiv preprint dg-ga/9508013} (1995).

\bibitem{vsevera2015poisson}
P.~{\v{S}}evera, {\it Poisson--lie t-duality and courant algebroids},  {\em
  Letters in Mathematical Physics} {\bf 105} (2015), no.~12 1689--1701.

\bibitem{getzler2010higher}
E.~Getzler, {\it Higher derived brackets},  {\em arXiv preprint
  arXiv:1010.5859} (2010).

\bibitem{fiorenza2007structures}
D.~Fiorenza and M.~Manetti, {\it L$_\infty$ structures on mapping cones},  {\em
  Algebra \& Number Theory} {\bf 1} (2007), no.~3 301--330.

\bibitem{Aldazabal:2013via}
G.~Aldazabal, M.~Gra{\~n}a, D.~Marqu{\'e}s, and J.~A. Rosabal, {\it {The gauge
  structure of Exceptional Field Theories and the tensor hierarchy}},  {\em
  JHEP} {\bf 1404} (2014) 049, [\href{http://arxiv.org/abs/1312.4549}{{\tt
  arXiv:1312.4549}}].

\bibitem{Wang:2015hca}
Y.-N. Wang, {\it {Generalized Cartan Calculus in general dimension}},  {\em
  JHEP} {\bf 07} (2015) 114, [\href{http://arxiv.org/abs/1504.04780}{{\tt
  arXiv:1504.04780}}].

\bibitem{Cederwall:2013naa}
M.~Cederwall, J.~Edlund, and A.~Karlsson, {\it {Exceptional geometry and tensor
  fields}},  {\em JHEP} {\bf 1307} (2013) 028,
  [\href{http://arxiv.org/abs/1302.6736}{{\tt arXiv:1302.6736}}].

\bibitem{alexandrov1997geometry}
M.~Alexandrov, A.~Schwarz, O.~Zaboronsky, and M.~Kontsevich, {\it The geometry
  of the master equation and topological quantum field theory},  {\em
  International Journal of Modern Physics A} {\bf 12} (1997), no.~07
  1405--1429.

\bibitem{Townsend:1995af}
P.~K. Townsend, {\it {D-branes from M-branes}},  {\em Phys. Lett.} {\bf B373}
  (1996) 68--75, [\href{http://arxiv.org/abs/hep-th/9512062}{{\tt
  hep-th/9512062}}]. [,120(1995)].

\bibitem{Tseytlin:1996it}
A.~A. Tseytlin, {\it {Selfduality of Born-Infeld action and Dirichlet
  three-brane of type IIB superstring theory}},  {\em Nucl. Phys.} {\bf B469}
  (1996) 51--67, [\href{http://arxiv.org/abs/hep-th/9602064}{{\tt
  hep-th/9602064}}].

\bibitem{Ikeda:2002wh}
N.~Ikeda, {\it {Chern-Simons gauge theory coupled with BF theory}},  {\em Int.
  J. Mod. Phys.} {\bf A18} (2003) 2689--2702,
  [\href{http://arxiv.org/abs/hep-th/0203043}{{\tt hep-th/0203043}}].

\bibitem{Kalkkinen:2002tk}
J.~Kalkkinen and K.~S. Stelle, {\it {Large gauge transformations in M theory}},
   {\em J. Geom. Phys.} {\bf 48} (2003) 100--132,
  [\href{http://arxiv.org/abs/hep-th/0212081}{{\tt hep-th/0212081}}].

\bibitem{Intriligator:2000eq}
K.~A. Intriligator, {\it {Anomaly matching and a Hopf-Wess-Zumino term in 6d,
  N=(2,0) field theories}},  {\em Nucl. Phys.} {\bf B581} (2000) 257--273,
  [\href{http://arxiv.org/abs/hep-th/0001205}{{\tt hep-th/0001205}}].

\bibitem{Palmer:2013pka}
S.~Palmer and C.~Sämann, {\it {Six-Dimensional (1,0) Superconformal Models and
  Higher Gauge Theory}},  {\em J. Math. Phys.} {\bf 54} (2013) 113509,
  [\href{http://arxiv.org/abs/1308.2622}{{\tt arXiv:1308.2622}}].

\bibitem{Lavau:2014iva}
S.~Lavau, H.~Samtleben, and T.~Strobl, {\it {Hidden Q-structure and Lie
  3-algebra for non-abelian superconformal models in six dimensions}},  {\em J.
  Geom. Phys.} {\bf 86} (2014) 497--533,
  [\href{http://arxiv.org/abs/1403.7114}{{\tt arXiv:1403.7114}}].

\bibitem{Deser:2014mxa}
A.~Deser and J.~Stasheff, {\it {Even symplectic supermanifolds and double field
  theory}},  {\em Commun. Math. Phys.} {\bf 339} (2015), no.~3 1003--1020,
  [\href{http://arxiv.org/abs/1406.3601}{{\tt arXiv:1406.3601}}].

\bibitem{Heller:2016abk}
M.~A. Heller, N.~Ikeda, and S.~Watamura, {\it {Unified picture of non-geometric
  fluxes and T-duality in double field theory via graded symplectic
  manifolds}},  {\em JHEP} {\bf 02} (2017) 078,
  [\href{http://arxiv.org/abs/1611.08346}{{\tt arXiv:1611.08346}}].

\bibitem{Chatzistavrakidis:2018ztm}
A.~Chatzistavrakidis, L.~Jonke, F.~S. Khoo, and R.~J. Szabo, {\it {Double Field
  Theory and Membrane Sigma-Models}},
  \href{http://arxiv.org/abs/1802.07003}{{\tt arXiv:1802.07003}}.

\bibitem{Chatzistavrakidis:2015vka}
A.~Chatzistavrakidis, L.~Jonke, and O.~Lechtenfeld, {\it {Sigma models for
  genuinely non-geometric backgrounds}},  {\em JHEP} {\bf 11} (2015) 182,
  [\href{http://arxiv.org/abs/1505.05457}{{\tt arXiv:1505.05457}}].

\bibitem{Fiorenza:2013nha}
D.~Fiorenza, H.~Sati, and U.~Schreiber, {\it {Super Lie n-algebra extensions,
  higher WZW models, and super p-branes with tensor multiplet fields}},  {\em
  Int. J. Geom. Meth. Mod. Phys.} {\bf 12} (2014) 1550018,
  [\href{http://arxiv.org/abs/1308.5264}{{\tt arXiv:1308.5264}}].

\bibitem{Townsend:1997wg}
P.~K. Townsend, {\it {M theory from its superalgebra}},  in {\em {Strings,
  branes and dualities. Proceedings, NATO Advanced Study Institute, Cargese,
  France, May 26-June 14, 1997}}, pp.~141--177, 1997.
\newblock \href{http://arxiv.org/abs/hep-th/9712004}{{\tt hep-th/9712004}}.

\bibitem{DeserSaemannUnpublushed}
A.~Deser and C.~Sämann, {\it \rm unpublished}, .

\bibitem{Cremmer:1979up}
E.~Cremmer and B.~Julia, {\it {The SO(8) Supergravity}},  {\em Nucl.Phys.} {\bf
  B159} (1979) 141.

\bibitem{Hitchin:2004ut}
N.~Hitchin, {\it {Generalized Calabi-Yau manifolds}},  {\em Quart.J.Math.Oxford
  Ser.} {\bf 54} (2003) 281--308,
  [\href{http://arxiv.org/abs/math/0209099}{{\tt math/0209099}}].

\bibitem{Ashmore:2015joa}
A.~Ashmore and D.~Waldram, {\it {Exceptional Calabi-Yau spaces: the geometry of
  $\mathcal{N}=2$ backgrounds with flux}},  {\em Fortsch. Phys.} {\bf 65}
  (2017), no.~1 1600109, [\href{http://arxiv.org/abs/1510.00022}{{\tt
  arXiv:1510.00022}}].

\bibitem{Hull:2004in}
C.~Hull, {\it {A Geometry for non-geometric string backgrounds}},  {\em JHEP}
  {\bf 0510} (2005) 065, [\href{http://arxiv.org/abs/hep-th/0406102}{{\tt
  hep-th/0406102}}].

\bibitem{Hull:2006va}
C.~M. Hull, {\it {Doubled Geometry and T-Folds}},  {\em JHEP} {\bf 0707} (2007)
  080, [\href{http://arxiv.org/abs/hep-th/0605149}{{\tt hep-th/0605149}}].

\bibitem{deBoer:2010ud}
J.~de~Boer and M.~Shigemori, {\it {Exotic branes and non-geometric
  backgrounds}},  {\em Phys.Rev.Lett.} {\bf 104} (2010) 251603,
  [\href{http://arxiv.org/abs/1004.2521}{{\tt arXiv:1004.2521}}].

\bibitem{Gualtieri:2003dx}
M.~Gualtieri, {\it {Generalized complex geometry}},
  \href{http://arxiv.org/abs/math/0401221}{{\tt math/0401221}}.

\bibitem{Hull:2009mi}
C.~Hull and B.~Zwiebach, {\it {Double Field Theory}},  {\em JHEP} {\bf 0909}
  (2009) 099, [\href{http://arxiv.org/abs/0904.4664}{{\tt arXiv:0904.4664}}].

\bibitem{Arvanitakis:2018hfn}
A.~S. Arvanitakis and C.~D.~A. Blair, {\it {The Exceptional Sigma Model}},
  \href{http://arxiv.org/abs/1802.00442}{{\tt arXiv:1802.00442}}.

\bibitem{Arvanitakis:2016zes}
A.~S. Arvanitakis and C.~D.~A. Blair, {\it {Black hole thermodynamics, stringy
  dualities and double field theory}},  {\em Class. Quant. Grav.} {\bf 34}
  (2017), no.~5 055001, [\href{http://arxiv.org/abs/1608.04734}{{\tt
  arXiv:1608.04734}}].

\bibitem{Hohm:2013nja}
O.~Hohm and H.~Samtleben, {\it {Gauge theory of Kaluza-Klein and winding
  modes}},  {\em Phys.Rev.} {\bf D88} (2013) 085005,
  [\href{http://arxiv.org/abs/1307.0039}{{\tt arXiv:1307.0039}}].

\bibitem{Hull:2014mxa}
C.~M. Hull, {\it {Finite Gauge Transformations and Geometry in Double Field
  Theory}},  {\em JHEP} {\bf 04} (2015) 109,
  [\href{http://arxiv.org/abs/1406.7794}{{\tt arXiv:1406.7794}}].

\bibitem{Alexandrov:1995kv}
M.~Alexandrov, A.~Schwarz, O.~Zaboronsky, and M.~Kontsevich, {\it {The Geometry
  of the master equation and topological quantum field theory}},  {\em Int. J.
  Mod. Phys.} {\bf A12} (1997) 1405--1429,
  [\href{http://arxiv.org/abs/hep-th/9502010}{{\tt hep-th/9502010}}].

\bibitem{Howe:1989uk}
P.~S. Howe and P.~K. Townsend, {\it {CHERN-SIMONS QUANTUM MECHANICS}},  {\em
  Class. Quant. Grav.} {\bf 7} (1990) 1655--1668.

\bibitem{severa2001some}
P.~{\v{S}}evera, {\it Some title containing the words" homotopy" and"
  symplectic", eg this one},  {\em arXiv preprint math/0105080} (2001).

\bibitem{Coimbra:2011ky}
A.~Coimbra, C.~Strickland-Constable, and D.~Waldram, {\it {$E_{d(d)} \times
  \mathbb{R}^+$ generalised geometry, connections and M theory}},  {\em JHEP}
  {\bf 02} (2014) 054, [\href{http://arxiv.org/abs/1112.3989}{{\tt
  arXiv:1112.3989}}].

\bibitem{Keir:2013jga}
J.~Keir, {\it {Stability, Instability, Canonical Energy and Charged Black
  Holes}},  {\em Class. Quant. Grav.} {\bf 31} (2014), no.~3 035014,
  [\href{http://arxiv.org/abs/1306.6087}{{\tt arXiv:1306.6087}}].

\bibitem{Aharony:1996wp}
O.~Aharony, {\it {String theory dualities from M theory}},  {\em Nucl. Phys.}
  {\bf B476} (1996) 470--483, [\href{http://arxiv.org/abs/hep-th/9604103}{{\tt
  hep-th/9604103}}].

\bibitem{Pasti:1997gx}
P.~Pasti, D.~P. Sorokin, and M.~Tonin, {\it {Covariant action for a D = 11
  five-brane with the chiral field}},  {\em Phys. Lett.} {\bf B398} (1997)
  41--46, [\href{http://arxiv.org/abs/hep-th/9701037}{{\tt hep-th/9701037}}].

\bibitem{Bergshoeff:2006gs}
E.~A. Bergshoeff, M.~de~Roo, S.~F. Kerstan, T.~Ortin, and F.~Riccioni, {\it
  {SL(2,R)-invariant IIB Brane Actions}},  {\em JHEP} {\bf 02} (2007) 007,
  [\href{http://arxiv.org/abs/hep-th/0611036}{{\tt hep-th/0611036}}].

\bibitem{Cederwall:1997ab}
M.~Cederwall and A.~Westerberg, {\it {World volume fields, SL(2:Z) and duality:
  The Type IIB three-brane}},  {\em JHEP} {\bf 02} (1998) 004,
  [\href{http://arxiv.org/abs/hep-th/9710007}{{\tt hep-th/9710007}}].

\bibitem{Douglas:1995bn}
M.~R. Douglas, {\it {Branes within branes}},  in {\em {Strings, branes and
  dualities. Proceedings, NATO Advanced Study Institute, Cargese, France, May
  26-June 14, 1997}}, pp.~267--275, 1995.
\newblock \href{http://arxiv.org/abs/hep-th/9512077}{{\tt hep-th/9512077}}.

\bibitem{Deser:2018oyg}
A.~Deser and C.~Sämann, {\it {Derived Brackets and Symmetries in Generalized
  Geometry and Double Field Theory}},  in {\em {17th Hellenic School and
  Workshops on Elementary Particle Physics and Gravity (CORFU2017) Corfu,
  Greece, September 2-28, 2017}}, 2018.
\newblock \href{http://arxiv.org/abs/1803.01659}{{\tt arXiv:1803.01659}}.

\bibitem{Severa:2017oew}
P.~Ševera, {\it {Letters to Alan Weinstein about Courant algebroids}},
  \href{http://arxiv.org/abs/1707.00265}{{\tt arXiv:1707.00265}}.

\bibitem{Horava:1996ma}
P.~Horava and E.~Witten, {\it {Eleven-dimensional supergravity on a manifold
  with boundary}},  {\em Nucl. Phys.} {\bf B475} (1996) 94--114,
  [\href{http://arxiv.org/abs/hep-th/9603142}{{\tt hep-th/9603142}}].

\bibitem{Horava:1995qa}
P.~Horava and E.~Witten, {\it {Heterotic and type I string dynamics from
  eleven-dimensions}},  {\em Nucl. Phys.} {\bf B460} (1996) 506--524,
  [\href{http://arxiv.org/abs/hep-th/9510209}{{\tt hep-th/9510209}}].
  [,397(1995)].

\bibitem{Bergshoeff:2006bs}
E.~A. Bergshoeff, G.~W. Gibbons, and P.~K. Townsend, {\it {Open M5-branes}},
  {\em Phys. Rev. Lett.} {\bf 97} (2006) 231601,
  [\href{http://arxiv.org/abs/hep-th/0607193}{{\tt hep-th/0607193}}].

\bibitem{Hull:1997kt}
C.~M. Hull, {\it {Gravitational duality, branes and charges}},  {\em Nucl.
  Phys.} {\bf B509} (1998) 216--251,
  [\href{http://arxiv.org/abs/hep-th/9705162}{{\tt hep-th/9705162}}].

\bibitem{Bergshoeff:1998bs}
E.~Bergshoeff and J.~P. van~der Schaar, {\it {On M nine-branes}},  {\em Class.
  Quant. Grav.} {\bf 16} (1999) 23--39,
  [\href{http://arxiv.org/abs/hep-th/9806069}{{\tt hep-th/9806069}}].

\bibitem{grutzmann2011h}
M.~Gr{\"u}tzmann, {\it H-twisted lie algebroids},  {\em Journal of Geometry and
  Physics} {\bf 61} (2011), no.~2 476--484.

\bibitem{Ikeda:2010vz}
N.~Ikeda and K.~Uchino, {\it {QP-Structures of Degree 3 and 4D Topological
  Field Theory}},  {\em Commun. Math. Phys.} {\bf 303} (2011) 317--330,
  [\href{http://arxiv.org/abs/1004.0601}{{\tt arXiv:1004.0601}}].

\bibitem{Kokenyesi:2018ynq}
Z.~Kokenyesi, A.~Sinkovics, and R.~J. Szabo, {\it {AKSZ Constructions for
  Topological Membranes on $G_2$-Manifolds}},  {\em Fortsch. Phys.} {\bf 66}
  (2018), no.~3 1800018, [\href{http://arxiv.org/abs/1802.04581}{{\tt
  arXiv:1802.04581}}].

\bibitem{Severa:2016prq}
P.~Ševera, {\it {Poisson-Lie T-duality as a boundary phenomenon of
  Chern-Simons theory}},  {\em JHEP} {\bf 05} (2016) 044,
  [\href{http://arxiv.org/abs/1602.05126}{{\tt arXiv:1602.05126}}].

\bibitem{Sezgin:1998tm}
E.~Sezgin and P.~Sundell, {\it {Aspects of the M5-brane}},  in {\em
  {Nonperturbative aspects of strings, branes and supersymmetry. Proceedings,
  Spring School on nonperturbative aspects of string theory and supersymmetric
  gauge theories and Conference on super-five-branes and physics in 5 + 1
  dimensions, Trieste, Italy, March 23-April 3, 1998}}, pp.~369--389, 1998.
\newblock \href{http://arxiv.org/abs/hep-th/9902171}{{\tt hep-th/9902171}}.

\bibitem{Green:1996bh}
M.~B. Green, C.~M. Hull, and P.~K. Townsend, {\it {D-brane Wess-Zumino actions,
  t duality and the cosmological constant}},  {\em Phys. Lett.} {\bf B382}
  (1996) 65--72, [\href{http://arxiv.org/abs/hep-th/9604119}{{\tt
  hep-th/9604119}}].

\bibitem{jozwikowski2016note}
M.~J{\'o}{\'z}wikowski and M.~Rotkiewicz, {\it A note on actions of some
  monoids},  {\em Differential Geometry and its Applications} {\bf 47} (2016)
  212--245.

\bibitem{Ikeda:2012pv}
N.~Ikeda, {\it {Lectures on AKSZ Sigma Models for Physicists}},  in {\em
  {Proceedings, Workshop on Strings, Membranes and Topological Field Theory}},
  pp.~79--169, WSPC, WSPC, 2017.
\newblock \href{http://arxiv.org/abs/1204.3714}{{\tt arXiv:1204.3714}}.

\bibitem{Gruetzmann:2014ica}
M.~Grützmann and T.~Strobl, {\it {General Yang–Mills type gauge theories for
  $p$-form gauge fields: From physics-based ideas to a mathematical framework
  or From Bianchi identities to twisted Courant algebroids}},  {\em Int. J.
  Geom. Meth. Mod. Phys.} {\bf 12} (2014) 1550009,
  [\href{http://arxiv.org/abs/1407.6759}{{\tt arXiv:1407.6759}}].

\bibitem{Polchinski:1985zf}
J.~Polchinski, {\it {Evaluation of the One Loop String Path Integral}},  {\em
  Commun. Math. Phys.} {\bf 104} (1986) 37.

\end{thebibliography}\endgroup
\end{document}